\documentclass[showpacs,twocolumn,pra,longbibliography,superscriptaddress,notitlepage]{revtex4-2}
\usepackage{amsmath,amssymb,amsthm,mathrsfs,amsfonts,dsfont}
\usepackage{mathtools}
\usepackage{subfigure, epsfig}
\usepackage{braket}
\usepackage{bm}
\usepackage{enumerate}
\usepackage{color}
\usepackage[normalem]{ulem}
\usepackage{comment}
\usepackage[colorlinks = true]{hyperref}

\hypersetup{colorlinks=true, linkcolor=blue, citecolor=blue, urlcolor=black}

\usepackage{enumitem}

\usepackage{algorithm}  
\usepackage{algpseudocode}  

\graphicspath{{./figure/}}

\newtheorem{theorem}{Theorem}
\newtheorem{proposition}{Proposition}

\newtheorem{lemma}{Lemma}

\newtheorem{definition}{Definition}
\newtheorem{assumption}{Assumption}

\DeclareMathOperator*{\argmax}{arg\,max}

\newcommand{\mc}{\mathcal}

\newcommand{\mbb}{\mathbb}
\newcommand{\mr}{\mathrm}

\newcommand{\sinc}{\mathrm{sinc}}
\newcommand{\sech}{\mathrm{sech}}

\newcommand{\red}[1]{{\color{red} #1}}

\newcommand{\peking}{Center on Frontiers of Computing Studies, Peking University, Beijing 100871, China}

\newcommand{\oxford}{Clarendon Laboratory, University of Oxford, Oxford OX1 3PU, UK}
\newcommand{\chicago}{Pritzker School of Molecular Engineering, The University of Chicago, Illinois 60637, USA}

\begin{document}

\title{Universal quantum algorithmic  cooling on a quantum computer}

\begin{abstract}
Quantum cooling, a deterministic process that drives any state to the lowest eigenstate, has been widely used from studying ground state properties of chemistry and condensed matter quantum physics to general optimization problems. However, the cooling procedure is non-unitary, making its practical realization on a quantum computer a non-trivial task. 
Here, we propose universal quantum cooling algorithms with small resource costs. By utilizing a dual-phase representation of decaying functions, we show how to universally and deterministically realize a general cooling procedure with shallow quantum circuits. We demonstrate its applications in cooling an arbitrary input state with known ground state energy, corresponding to satisfactory, linear algebra tasks, and quantum state compiling tasks, and preparing eigenstates with unknown eigenenergies, corresponding to quantum many-body problems. Compared to quantum phase estimation, our method uses only one ancillary qubit, showing exponential improvement of the circuit complexity with respect to the final state infidelity while maintaining the Heisenberg limit of the eigenenergy estimation with proper cooling functions. We numerically benchmark the algorithms for the $8$-qubit Heisenberg model and verify its feasibility for accurately finding eigenenergies and obtaining eigenstate measurements. Our work paves the way for efficient and universal quantum algorithmic cooling with near-term and universal fault-tolerant quantum devices. 
\end{abstract}
\date{\today}
            
\author{Pei Zeng}
\email{peizeng.phy@gmail.com}   
\affiliation{\chicago}

\author{Jinzhao Sun}
\email{jinzhao.sun@physics.ox.ac.uk}
\affiliation{\peking}
\affiliation{\oxford}

\author{Xiao Yuan}
\email{xiaoyuan@pku.edu.cn}
\affiliation{\peking}
\affiliation{School of Computer Science, Peking University, Beijing 100871, China}
\maketitle


\emph{\textbf{Introduction.---}}
Many tasks, whether finding ground states~\cite{Sethground99,Alan05Science}, studying low temperature physics~\cite{Temme2011,Yung754}, solving optimization problems~\cite{farhi2014quantum,Lukin20prx}, or even simulating quantum dynamics~\cite{Kassal18681}, can be abstracted as a cooling problem. Among the various task-specific classical algorithms~---~such as perturbation theories~\cite{abrikosov2012methods}, variational approaches~\cite{WFT1,Orus2019}, self-consistent embedding methods~\cite{DFT2,DMET2012}, machine learning~\cite{PauliNet,Ferminet}, etc~---~imaginary time evolution~\cite{QMC} defines a natural and universal cooling procedure. Consider a Hamiltonian $ H$, the imaginary time evolution $e^{-\tau H}$ with a real-valued time $\tau$  monotonically lower down the average energy and deterministically drives an arbitrary pure state to the lowest eigenstate that has a nonzero overlap. Despite being mathematically universal, realizing imaginary time evolution for an arbitrary Hamiltonian is by no means an easy problem. The notorious sign problem~\cite{Signproblem} in Monte Carlo implementation of imaginary time evolution has limited its usage for solving general strongly correlated problems~\cite{Koloren__2011,Austin2012}. 

Can we realize a cooling procedure more efficiently on a quantum computer? The answer is yes! The most prominent algorithm is quantum phase estimation~\cite{kitaev1995quantum,kitaev2002classical}, which efficiently projects an arbitrary input quantum state to an eigenstate of the Hamiltonian, given the input state has a nonvanishing overlap with the target eigenstate. 
However, quantum phase estimation generally requires a deep circuit that is only realizable with a universal quantum computer~\cite{PhysRevX.8.041015,Gidney_2021}. With near-term noisy intermediate-scale quantum devices, variational approaches have been developed to emulate the imaginary time evolution with fixed or adaptive parametrized quantum circuits~\cite{McArdle2019,Yuan_2019,Motta2020,Nishi2021,amaro2021filtering}. Nevertheless, whether such circuits exist or how to efficiently find them still remains unclear~\cite{Abbas2021,holmes2021connecting,PhysRevLett.126.190501}. Several recent works have indicated the possibility of more efficient cooling algorithms~\cite{Ge19,lu2021algorithms,CVLCU,Rodeo21}, 
whereas they still assume the usage of many ancillary qubits or unconventional computational models.

 
Here, we propose a new paradigm for simulating general quantum cooling procedures on a quantum computer. For a general decaying function $g(h)$, such as $g = e^{-|h|}$ or $g=e^{-h^2}$, we show how to realize a cooling procedure $g(\tau H)$ for Hamiltonian $ H$ using solely real-time evolution and proper classical post-processing. The algorithm essentially implements general non-Hermitian dynamics, and it can efficiently find energy spectra and corresponding eigenstates, for any initial states that have nonvanishing overlaps with the target state. Compared to quantum phase estimation, our method only needs an exponentially shallower circuit with only one ancillary qubit and maintains the Heisenberg limit of the eigenenergy estimations when using proper cooling functions. Compared to variational approaches, our algorithm works deterministically without any assumption on parametrized circuits. \\



\emph{\textbf{Quantum algorithmic cooling.---}}
We start with a real-valued function $g(h)$ satisfying strictly non-increasing absolute value, $|g(h')|< |g(h)|,~\forall h'> h>0$ or $h'<h<0$~\footnote{The framework also works even if the function is not strictly non-increasing.}, and vanishing asymptotic value, 
$\lim_{\tau\rightarrow\infty}|g(\tau h')/g(\tau h)|=0,~\forall h'> h>0$ or $h'<h<0$ or
$\lim_{\tau\rightarrow\infty}|g(\tau h)/g(0)|=0,~\forall |h|> 0$. Consider a positive Hamiltonian $H=\sum_i E_i\ket{u_i}\bra{u_i}$ with eigenstates $\{\ket{u_i}\}$ and eigenenergies $E_i\ge 0$~\footnote{For any finite dimensional Hamiltonian, we can always shift it by a constant value to make it positive.}, it defines a cooling procedure $g(\tau H)$ with Hamiltonian $H$ and real-valued time $\tau$ for an arbitrary input state $\ket{\psi_0}=\sum_i c_i \ket{u_i}$ as
\begin{equation}\label{Eq:coolingdef}
	\ket{\psi(\tau)} = \frac{g(\tau H)\ket{\psi_0} }{\|g(\tau H)\ket{\psi_0} \|} = \frac{\sum_i g(\tau E_i)c_i\ket{u_i}}{\sum_i p_i g(\tau E_i)^2},
\end{equation}
with $p_i=|c_i|^2$. 
Since for larger eigenenergies the function $g(\tau E_i)$ decreases faster with $\tau$, the amplitudes of the normalized state $\ket{\psi(\tau)}$ concentrate to lower eigenstates, and the evolved state asymptotically approximates the lowest eigenstate $\ket{u_i}$ with nonzero $c_i$ for sufficiently large $\tau$. 


To realize the cooling function $g(h)$ with valid Fourier transforms, we consider the dual form of $g(h)$ via its Fourier transform $g(h) = \frac{1}{2\pi}\int^\infty_{-\infty} f(x) e^{ixh} \textrm dx$, with the inverse transform $f(x) = \int^\infty_{-\infty} g(h) e^{-ixh} \textrm dh$. Given the norm of $f(x)$ defined as $\|f\|=\int^\infty_{-\infty} |f(x)|\textrm dx$, we consider the normalized function $p(x)=|f(x)|/\|f\|$ and we have
\begin{equation} \label{eq:ghFourier}
	g(h) = c\int^\infty_{-\infty} e^{i\theta_x}p(x) e^{ixh} \textrm dx,
\end{equation}
where $c = {\|f\|}/{2\pi}$ and $e^{i\theta_x}=f(x)/|f(x)|$. 
The function $g(h)$ is \emph{realizable} if the following conditions hold.

\vspace{0.2 cm}
\noindent \emph{C1}: The normalization $\|f\|$ or $c$ is finite. 

\vspace{0.2 cm}

\noindent {\emph{C2}: $\big|1-\int^{L(\varepsilon)}_{-L(\varepsilon)} p(x) dx\big|\le \varepsilon, ~\forall \varepsilon\ge 0,~ \exists L(\varepsilon)=\mathcal O(\textrm{poly}(\frac{1}{\varepsilon}))$.}	

\vspace{0.2 cm}

\noindent The first condition implies a finite norm of $\|f\|$ and the second condition implies that a finite frequency in $[-L(\varepsilon), L(\varepsilon)]$ with $L(\varepsilon)=\mathcal O(\textrm{poly}(1/\varepsilon))$ is sufficient for approximating the normalized function $g(h)/c$, i.e., $\big|g(h)/c-\int^{L(\varepsilon)}_{-L(\varepsilon)} p(x) e^{ixh+\theta_x} \textrm dx\big|\le \varepsilon$.  In Table~\ref{Table:functinos}, we give several examples of $g(h)$.  Since our examples have zero phase, for a simpler presentation, we set the phase terms to be $0$ in the following discussion. Although, our following discussion naturally generalizes to the case with nonzero phases.  The detailed analysis of the cooling functions is given in Supplementary Materials. 

With the normalized dual phase representation $p(x)$ of $g(h)$, the cooling procedure of Eq.~\eqref{Eq:coolingdef} becomes
\begin{equation}\label{Eq:superposed}
	\ket{\psi(\tau)}  \propto \int^\infty_{-\infty} \textrm dp(x) \ket{\phi(x\tau)} ,
\end{equation}
which is now a superposition of real-time evolved states $\ket{\phi(x\tau)} = e^{ix\tau H}\ket{\psi_0}$ with amplitudes $\textrm dp(x) = p(x)\textrm dx$. The linear-combination-of-unitary algorithm~\cite{Long2011,LCU15,CVLCU} could be exploited to create the superposition.
However, creating such continuous-variable states is challenging with a digital quantum computer. Here, we propose an experimentally more feasible and efficient way to effectively realize the superposition.\\

\begin{table}[t]
\caption{Representative cooling functions $g(h)$. Here $f(x)$ corresponds to Fourier transform of $g(h)$, \emph{C1} refers to $c=\|f\|/2\pi$  which should be finite,  \emph{C2} refers to $L(\varepsilon)$ which should be less than $O(\textrm{poly}(\frac{1}{\varepsilon}))$, $\eta(\textrm{true/false})=1/0$, $\textrm{sinc}(x)=\frac{\sin (x)}{x}$, $\textrm{sech}(x)=\frac{2}{e^x+e^{-x}}$. Except for the rectangular function, the other four functions satisfy \emph{C1} and \emph{C2}.  }
\begin{tabular}{ccccc}
\hline
Type & $g(h)$ & $f(x)$ & \emph{C1} & \emph{C2}\\ \hline
Rectangular & $\eta(|h|\le 1/2)$ & $\textrm{sinc}(x/2\pi)$ & $\infty$ & $\infty$ \\
Triangle & $(1-|h|)\eta(|h|\le 1)$ & $\textrm{sinc}^2(x/2\pi)$ & $2\pi$&$6/\varepsilon$\\
Exponential & $e^{-|h|}$ & $\frac{2}{x^2+1}$ & $1$ & $2/(\pi\varepsilon)$ \\
Gaussian & $e^{-h^2}$ & $\sqrt{\pi}e^{-x^2/4}$ & $1$ & $2\sqrt{\ln(1/\varepsilon)}$ \\
Hyperbolic & $\textrm{sech}(h)$ & $\pi\textrm{sech}(\pi x/2)$ & $1$ & $\frac{2}{\pi}\ln(1/\varepsilon)$ \\
\hline
\end{tabular}
\label{Table:functinos}
\end{table}

Instead of preparing the quantum state of Eq.~\eqref{Eq:superposed}, we focus on the goal to obtain arbitrary observable expectation values.  
Specifically, we aim to measure any observable $O$ of the evolved state $\ket{\psi(\tau)}$, i.e., 
\begin{equation}\label{Eq:OND}
	\braket{O}_{\psi(\tau)} = \braket{\psi(\tau)|O|\psi(\tau)} = \frac{N_\tau(O) }{D_\tau},
\end{equation}
where $N_\tau(O)=\int^\infty_{-\infty}\int^\infty_{-\infty} \textrm dp(x,x') \braket{\phi(x'\tau)|O|\phi(x\tau )} $ and $D_\tau=\int^\infty_{-\infty}\int^\infty_{-\infty} \textrm dp(x,x') \braket{\phi(x'\tau)|\phi(x\tau )} $ with $\textrm dp(x,x') = p(x)p(x')\textrm dx\textrm dx'$. Here $D_\tau$ can be simplified as $D_\tau=\int^\infty_{-\infty} \textrm d\tilde p(y) \braket{\psi_0|e^{iy\tau H}|\psi_0} $ with $\textrm d\tilde p(y) = \tilde p(y)\textrm dy$ and $\tilde p(y) =\frac{1}{2}\int^\infty_{-\infty}p(\frac{z+y}{2})p(\frac{z-y}{2})\textrm dz$. 
For the numerator $N_\tau(O)$, we can efficiently obtain it by sampling the distribution $\textrm dp(x,x')$ and then estimating the mean value $\mathbb E_{x,x'} \braket{\phi(x'\tau)|O|\phi(x\tau )}$, where each term is realizable on a quantum computer with the Hadamard test circuit~\cite{supplementary}. We can similarly obtain the denominator $D_\tau$ by estimating $ \mathbb E_y \braket{\psi_0|e^{iy\tau H}|\psi_0}$ with probability $\textrm d\tilde p(y)$. Therefore, we only need a quantum computer to efficiently estimate the values of state overlaps like $\braket{\psi_0|U|\psi_0}$ with $U$ being either $e^{-ix'\tau H}Oe^{ix\tau H}$ or $e^{iy\tau H}$, and then we can effectively obtain the time-dependent expectation value of any observable by post-processing the measurement outcomes. In practice, we also consider a cutoff of the integral from $[-\infty,\infty]$ to $[-x_m,x_m]$ to avoid infinite integration. We show shortly the analysis on the circuit and sample complexity of the algorithm. \\

\emph{\textbf{Applications.---}}Here, we
consider applications of quantum algorithmic cooling for quantum systems with known and unknown eigenenergies. The detailed algorithms can be found in Supplementary Materials.

First, we consider Hamiltonians $H$ with a known ground state energy $E_0\ge 0$, corresponding to satisfactory problems~\cite{farhi2014quantum,Lukin20prx}, linear algebra tasks~\cite{bravo-prieto2019, XU2021, huang2019near}, quantum state compiling~\cite{QAQC,heya2018variational,jones2018quantum,sharma2019noise}, quantum error correction~\cite{johnson2017qvector,xu2019variational}, etc. 
For an initial state $\ket{\psi_0}=\sum_ic_i\ket{u_i}$ with a nonvanishing $p_0 =|\alpha_0|^2$, we can obtain the average of any observable of a well approximated the ground state $\ket{u_0}$ according to Eq.~\eqref{Eq:OND}.
However,  when $E_0>0$, the denominator
$
	D_\tau = {\braket{\psi_0|g(\tau H)^2|\psi_0}}/c^2
$,
may decrease exponentially with $\tau$, where $c=(2\pi/\|f\|)^2$ is a constant. For example, when $g(h)=e^{-|h|}$, we have $D_\tau=\sum_ip_ie^{-2\tau E_i} < p_0e^{-2\tau E_0}$ with $p_i=|c_i|^2$. In this case, the numerator also decreases exponentially with $\tau$ and we need an exponential number of samples to achieve a desired accuracy~\footnote{We note that the argument assumes a large $\tau$ independent of the measurement accuracy. However, when the gap between the ground and the first excited state is finite, as happens in satisfactory problems, linear algebra tasks, and quantum state compiling, we only need to choose $\tau=\mathcal O(\log(1/\varepsilon))$ as a logarithmic function of the state infidelity $\varepsilon$ (when $g(h)$ is exponential, Gaussian, or hyperbolic). Then the numerator $N_\tau(O)$ and denominator $D_\tau$ scales  as $\mathcal O(\textrm{poly}(1/\varepsilon))$ and the sample complexity will not be exponentially large.}. 
A simple strategy to circumvent the exponential decay is to shift the Hamiltonian by $-E_0$, so that the shifted ground state energy is 0, and  for large $\tau$, the denominator $D_\tau$ scales as $p_0$. Furthermore, it is not hard to see that for the valid cooling functions considered in Table~\ref{Table:functinos}, the algorithm works for any eigenstate $\ket{u_i}$ with known eigenenergy $E_i$, by shifting the Hamiltonian by $-E_i$. The algorithm is efficient as long as the probability $p_i$ is not too small. 

For the second task, we consider problems with unknown (inaccurate) eigenenergy, such as chemistry or condensed matter problems~\cite{cao2018quantum,mcardle2020quantum,Bauer2020,motta2021emerging}. 
Starting with a guess $E$ of an eigenenergy (for instance, obtained from classical methods), we shift the Hamiltonian by $-E$ and estimate the denominator $D_{\tau,\psi_0}(E)$  as
$
		D_{\tau,\psi_0}(E) = {\braket{\psi_0|g((H-E)\tau)^2|\psi_0}}/c^2
$.
Note that $D_{\tau,\psi_0}(E)$ is locally maximized to $p_i/c^2$ when $E=E_i$ is an eigenenergy of $H$, otherwise it exponentially decays with $\tau$. Then we can scan the denominator $D_{\tau,\psi_0}(E)$ with different trial energy $E$ and find the eigenenergies of $H$ via the peaks of $D_{\tau,\psi_0}(E)$. Furthermore, the denominator $D_{\tau,\psi_0}(E)$ is 
\begin{equation}
	D_{\tau,\psi_0}(E)=\int^\infty_{-\infty} \textrm d\tilde p(y) e^{-iy\tau E}\braket{\psi_0|e^{iy\tau H}|\psi_0}, 
\end{equation}
where $\braket{\psi_0|e^{iy\tau H}|\psi_0}$ is obtained from quantum computers, and the information of $E$ and the distribution $\textrm d\tilde p(y)$ only appears as the coefficient. Therefore $D_{\tau,\psi_0}(E)$ with different $E$ could be obtained from different classical post-processing of the same measurement outcome $\braket{\psi_0|e^{iy\tau H}|\psi_0}$, making the sample-efficiency of the algorithm independent of the value and the number of trial energies $E$. Due to the local monotonic feature of $g$, there also exist efficient classical algorithms, such as binary search algorithms, to find the peaks.

In the above discussion, we have assumed a `good' initial state that has a nonvanishing overlap with the target state, an assumption usually adopted in quantum phase estimation~\cite{nielsen2002quantum} or other non-variational ground-state preparation algorithms~\cite{albash2018adiabatic,Ge19,lin2020nearoptimalground}. Although this assumption is hard to promise for arbitrary problems, it is reasonable in many practical scenarios.
For example, we could just use a classical approximated solution~\cite{WFT1,Orus2019}, such as the Hartree-Fock state or a tensor network solution, or we can apply adiabatic state preparation to enhance the state overlap~\cite{yuan2020quantum,Sethground99,Alan05Science}. Besides, one can use shallow parametrized quantum circuits and apply variational optimization to search for an approximated solution~\cite{Sugurureview,Cerezo2021,bharti2021noisy}. Integrating these approximation algorithms with our quantum algorithmic cooling, it provides an efficient way to study eigenstates and eigenenergies of many-body quantum systems. 
\\

\emph{\textbf{Error and resource estimation.---}}When applying the cooling algorithm in practice, we have to take the error caused by finite circuit depth and finite sampling number into consideration. There are three dominant factors limiting the precision of our estimation~---~the finite imaginary time $\tau$, the cutoff $x_m$ of the integration, and the finite sampling number $N_M$ during the experiment. We summarize how those factors affect  observable and eigenenergy estimation and refer to Supplementary Materials for details.


First, we consider the error of observable estimation with a known target eigenenergy. 
Denote the ideal (estimated) measurements with infinite (finite) $\tau$, $x_m$, and $N_M$ as $\braket{O}$ ($\braket{\hat{O}}_{\tau}^{(x_m)}$), the effect of the three factors on observable estimation is given as follows. 
\begin{theorem}[Accuracy of the observable estimation] \label{thm:observable} 
Given constant $K>0$, error $\varepsilon\in (0,1)$, finite imaginary time $\tau\geq g^{-1}\left({\varepsilon p_j}/{6}\right)/{\Delta}$, normalized cutoff time $x_m\geq \sqrt{2} L\left({\varepsilon p_j}/{12}\right)$, and sample number $N_M \geq K/ ({\varepsilon p_j}/{6})^{2}$, the error between the expectation value estimation $\braket{\hat{O}}^{(x_m)}_\tau$ and the ideal expectation value $\braket{O}$ is bounded by
$
|\braket{\hat{O}}_{\tau}^{(x_m)} - \braket{O} | \leq \varepsilon \|O\|_1$
with a failure probability of $\delta=4 \exp(-K/8)$. Here $\|O\|_1=\sum_l |o_l|$ where $\{o_l\}_l$ are the Pauli coefficients $O=\sum_l o_l P_l$, $\Delta=\min\{|E_{j-1}-E_j|,|E_{j+1}-E_j|\}$, and $p_j$ is the overlap with the target eigenstate.
\end{theorem}

Theorem~\ref{thm:observable} states that the asymptotic time complexity of the cooling algorithm is determined by $g^{-1}(\varepsilon)L(\varepsilon)$ of the cooling function. The valid cooling functions in Table~\ref{Table:functinos} all satisfy $g^{-1}(\varepsilon)=\mc{O}(\log(1/\varepsilon))$. Therefore, we can choose $\tau$ to be $\mc{O}( \log(1/(p_j\varepsilon_\tau))\Delta^{-1})$, which is logarithmic to the inverse state overlap $1/p_j$ and inverse error $1/\varepsilon_\tau$, and linear to the inverse gap $\Delta^{-1}$.
For the cutoff $x_m$, we have shown that for the triangle and exponential cooling functions, we have $L(\varepsilon)=\mathcal{O}(\textrm{poly}({1}/{\varepsilon}))$, and for the Gaussian and secant hyperbolic cooling functions, we can achieve even better results $\mathcal{O}(\sqrt{\log({1}/{\varepsilon})})$ or $L(\varepsilon)=\mathcal{O}(\log({1}/{\varepsilon}))$. Take the Gaussian function as an example, we have $x_m = \mathcal O(\sqrt{\log(1/(\varepsilon p_j))})$. Note that the maximum real-time evolution  is given by $t_m = \tau x_m$. 
As a result, we have the following result. 

\begin{proposition}
The time (circuit) $t_m$ and sample $N_M$ complexity for the Gaussian cooling function is 
\begin{equation}
\begin{aligned}
t_m &\sim  \mc{O}\big( \Delta^{-1}\log(1/(\varepsilon p_j)) \big), \\
N_M &\sim \mathcal O(1/(\varepsilon p_j)^{2}).
\end{aligned}
\end{equation}
\end{proposition}
\noindent The time or circuit complexity is logarithmic to $1/(p_j\varepsilon_\tau)$, which is exponentially better than that of quantum phase estimation or adiabatic state preparation, which is generally polynomial to $1/(p_j\varepsilon_\tau)$. The sample complexity is slightly worse than that of quantum phase estimation, which is $\mathcal O(1/( p_j\varepsilon^{2}))$. 
\\

Next, we consider the error of eigenenergy estimation. 
\begin{theorem}[Accuracy of the eigenenergy estimation] \label{thm:eigenenergy}
Given constant $K>0$, error $\varepsilon\in (0,1)$, finite imaginary time $\tau\geq {1}/{\kappa}$, normalized cutoff time $x_m\geq \sqrt{2}L\left({(1-g(1))}p_j/{4}\right)$, and sample number $N_M\geq {2K}p_j^{-2}/({1-g(1)})  $, the error between the estimated eigenenergy  $\hat{E}_j$ and the one ideal  $E_j$ is
$
|\hat{E}_j - E_j | \leq \kappa$, 
with a failure probability of $2\delta= 4\exp\left(-K/8 \right)$.
\end{theorem}
\noindent Note that since only the imaginary time $\tau$ depends on the inverse accuracy $1/\kappa$, the total cost $\tau\cdot x_m\cdot N_M$ (circuit complexity $\times$ sample complexity) also scales as $1/\kappa$, indicating Heisenberg's limit for eigenenergy estimation, {similar to a former result by Lin and Tong~\cite{lin2021heisenberglimited}}.
We refer to Table I and Table II in \cite{supplementary} for the complexity comparison in the eigenenergy estimation and observable estimation, respectively.
\\

\begin{figure}[t]
\centering
\includegraphics[width =0.87\linewidth]
{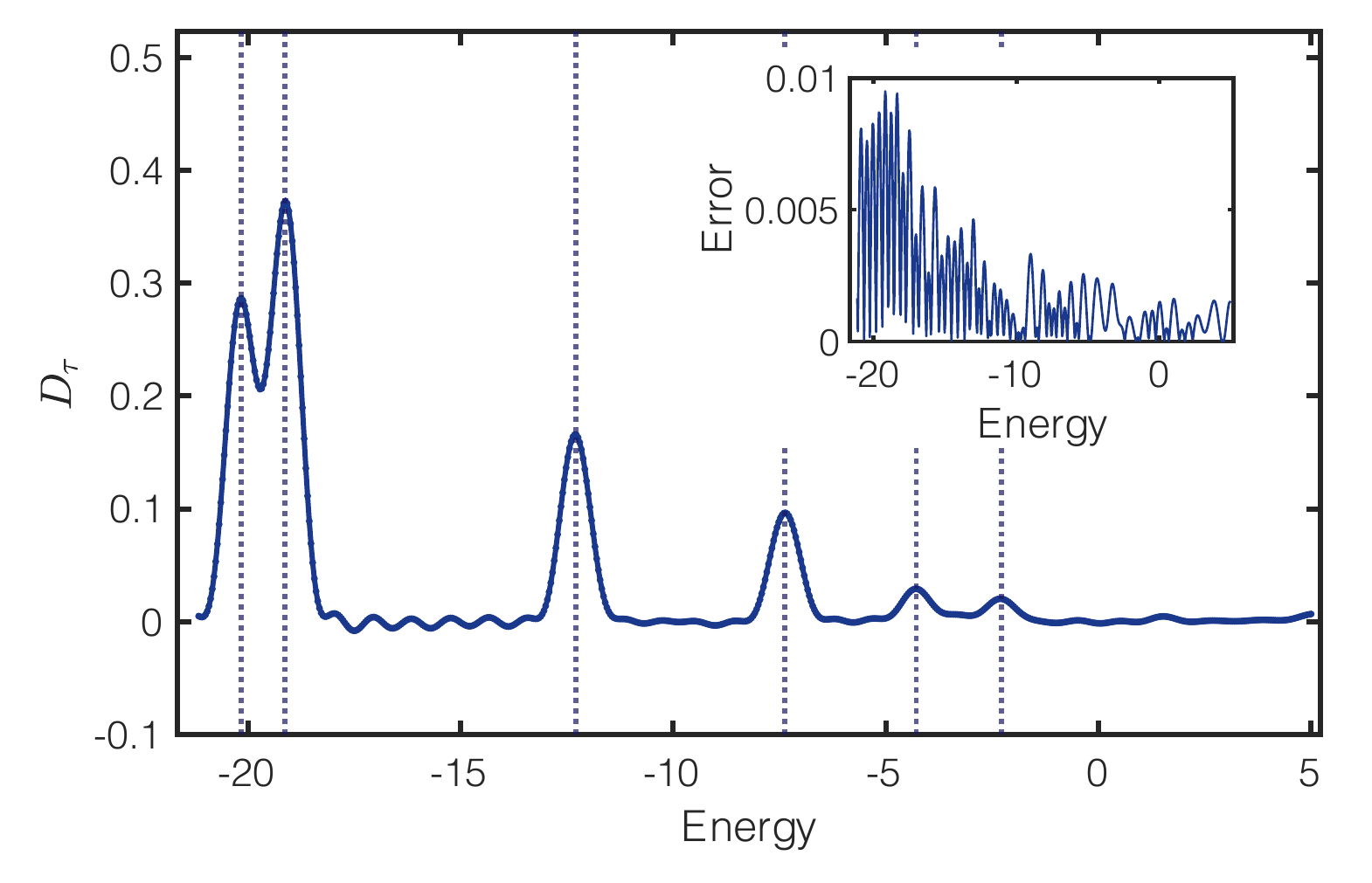}
 \caption{Energy spectra search of the $8$-site anisotropic Heisenberg model using the Gaussian cooling function. The solid line shows the denominator of $D_{\tau}$ with $\tau = 1.7$ and cutoff $x_m=4.4$.
 The dashed line shows the exact eigenenergy of the Heisenberg model. The figure inset shows the error of $D_{\tau}$ with respect to that under the exact cooling. We set $N_s = 10^5$ in the Monte Carlo sampling of the integral, and we use the expectation value for each sample that ignores measurement shot noise.
 }
\label{fig:spectra}
\end{figure}

\emph{\textbf{Numerical test.---}}Here we show numerical implementation of the  algorithmic cooling method for the anisotropic Heisenberg Hamiltonian,
$H = J\sum_i \left(\sigma_i^x \sigma_{i+1}^x + \sigma_i^y \sigma_{i+1}^y + 2 \sigma_i^z \sigma_{i+1}^z\right)  + h\sum_i \sigma_i^z$, where $J=1$ is the exchange coupling, $\sigma_i^{\alpha}$ ($\alpha=x,y,z$) is the Pauli operator on the $i$th
site, $h=1$ is strength of a uniform magnetic field in the $z$ direction, and we impose periodic boundary condition. We consider the initial state in the computation basis as $|01010101\rangle$, which is close to the ground state and has nonzero overlap with a relatively small number of low-energy eigenstates.
We first consider determining energy spectra 
by searching for the peaks of $D_{\tau}(E)$ with the  Gaussian cooling function. Using a finite imaginary time $\tau =1.7$ and cutoff $x_m=4.4$ according to $\tau,\, x_m \propto \sqrt{2\log(1/\epsilon)}$ with error $\varepsilon$, and $10^5$ number of samples for the integral, we show the energy spectra in Fig.~\ref{fig:spectra}. The maximum error introduced is below $0.01$, which aligns with our error analysis. We further show the searching process with different time $\tau$ in Ref.~\cite{supplementary}.

\begin{figure}[t]
\centering
\includegraphics[width =1.05\linewidth]
{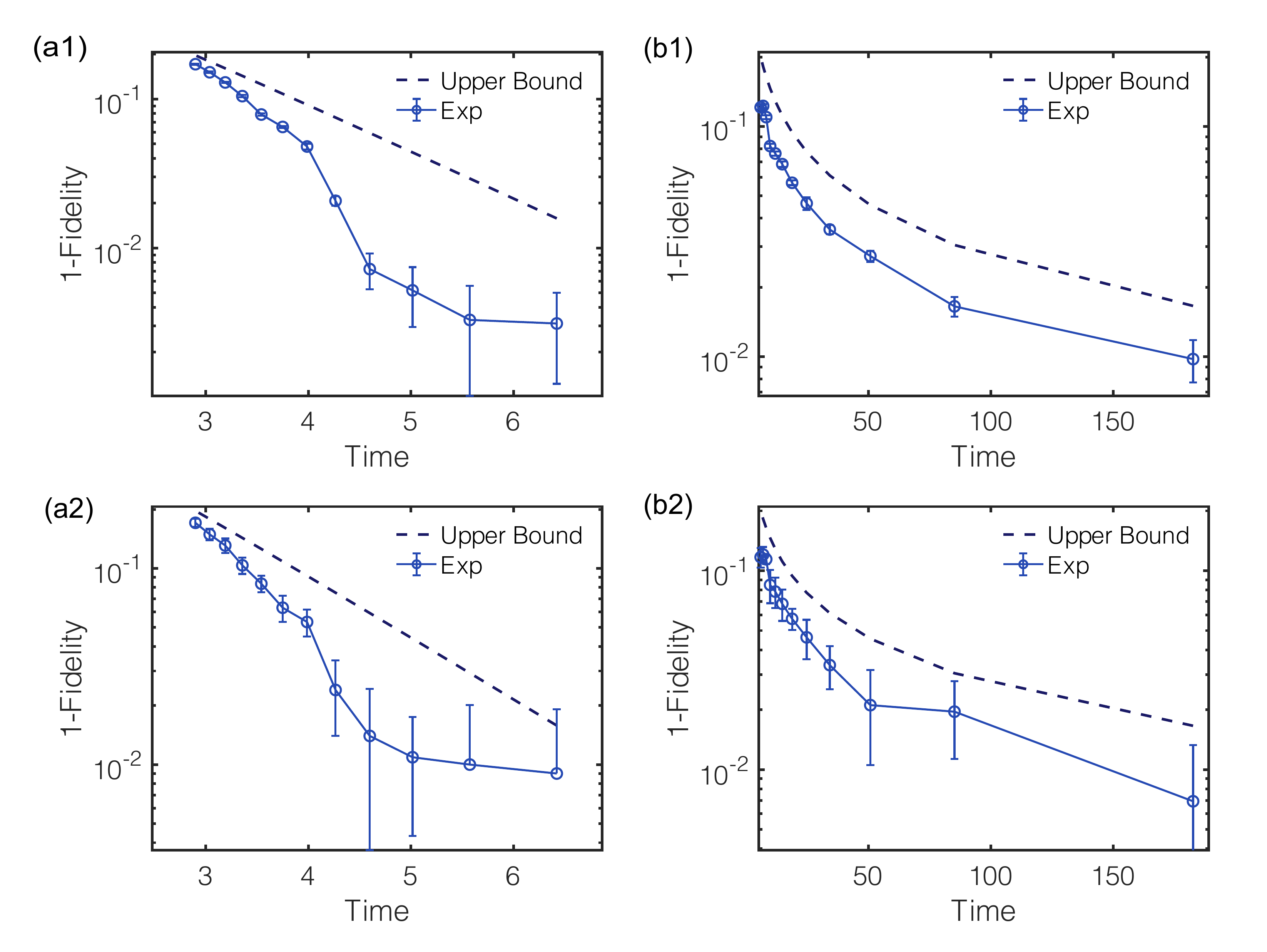}
 \caption{Error dependence of the total evolution time with (a1-a2) Gaussian  and (b1-b2) exponential cooling functions.  
In (a1) and (b1), we adjust both the imaginary time $\tau$ and the cutoff $x_m$, and study the infidelity versus maximal total time $t_m=\tau x_m$. We ignore the measurement shot noise in (a1,  b1).
In (a2) and (b2), we further consider measurement shot noise, which introduces more instability of the infidelity. Nevertheless, we still observe the exponential and {linear inverse} decay of the infidelity for the Gaussian and exponential functions, respectively. 
The errorbar is the standard deviation of $D_{\tau}$ error over 10 independent repetitions of the entire setup.
 We use $10^5$ samples in the Monte Carlo sampling of the integral.
Dashed line: theoretical upper bound of the infidelity assuming infinite sample;
Dots with error bar: realistic infidelity with a constant number of samples.
}
\label{fig:scale}
\end{figure}


Next, we show the error dependence of the eigenstate preparation with a special focus on the  time or circuit complexity. 
Suppose we aim to find the second eigenstate $\ket{u_2}$, which has the largest overlap with the initial state. Given the associated eigenenergy $E_2$ found by the above method, we now analyze the error introduced from finite imaginary time and cutoff.
Here, we focus on the state infidelity of the normalized state after cooling and the target ideal eigenstate, which can be expressed as $\varepsilon = 1-\langle O \rangle_{\tau}$ with $O = \ket{u_2}\bra{u_2}$. 
We show the error dependence on the imaginary time $\tau$ and the maximal total evolution time $t_m = \tau x_m$ with Gaussian and exponential cooling functions in Fig.~\ref{fig:scale}(a1-a2) and (b1-b2), respectively. We plot the theoretical upper bound of the infidelity in the dashed line in the figures (assuming finite $\tau$, infinite $x_m$, and infinite samples $N_M$) and realistic infidelity with a constant resources using dots.
In Fig~\ref{fig:scale}(a1) and (b1), we consider finite $\tau$ and $x_m$ with $10^5$ samples for calculating the integral, but ignore measurement errors. We see the exponential scaling with respect to the total evolution time $t_m$ for the Gaussian cooling function, while it asymptotically becomes  $1/\varepsilon$ for the exponential cooling function.
In Fig~\ref{fig:scale}(a2) and (b2), we further consider single-shot measurement outcomes again with a total of $10^5$ samples and show the error dependence accordingly. \\




\emph{\textbf{Conclusion \& Discussion.---}}We propose a universal and economic quantum algorithm to implement general cooling procedures. The algorithm is applicable to estimate eigenenergies and prepare eigenstates of problems with known and unknown eigenenergies when the overlap with the initial state is nonvanishing. In practice, we can use adiabatic state preparation or variational quantum eigensolver to prepare a good initial state. Our method is more competitive to quantum phase estimation or bare variational algorithms when the circuit depth and qubits are limited.


Another potential application is to study system properties under the finite-temperature condition. Given the ability of implementing general cooling function or preparing eigenstates, one may combine existing technique, such as Metropolis sampling~\cite{lu2021algorithms}, to effectively prepare thermal states and study finite-temperature physics. 
Besides, the algorithm may also help realize a general Markovian quantum process, which is described by a non-Hermitian Hamiltonian~\cite{PhysRevLett.125.010501}. We can decompose the Hamiltonian to the real and imaginary part, and further expand the imaginary part using Fourier transform. In this way, we could simulate a general non-Hermitian dynamic process.

In this work, we mainly explore the power of continuous Fourier transform in Eq.~\eqref{eq:ghFourier} to construct the cooling function $g(h)$. We remark that the Fourier transform is a universal tool to construct any non-unitary process else than cooling. For example, one can create a ``two-peak'' cooling function to prepare a cat-like state, i.e., the superposition of two remote eigenstates. Meanwhile, we can explore different linear transforms of continuous or discrete functions, such as discrete Fourier transform and continuous wavelet transform, to study their possible usage in realizing even more general non-unitary processes. Several related recent ideas include Refs.~\cite{Ge19,lu2021algorithms,CVLCU,Rodeo21,lin2021heisenberglimited}. Comparison and combination of these ideas are also an interesting future work.\\

\begin{acknowledgments}

This work is supported by the National Natural Science Foundation of China Grant No.~12175003.
The numerics is supported by High-performance Computing Platform of Peking University.

\emph{Note added.---}Recently, we became aware of a relevant work by \textcite{huo2021shallow}, who considered a specific cooling function using the Fourier transform of the Lorentz-Gaussian function. They studied its application for realizing imaginary time evolution and discussed its combination with error mitigation for  implementation with near-term quantum devices.
\end{acknowledgments}


%

\clearpage
\newpage
\widetext

\section*{Supplementary Material for ``Universal algorithmic quantum cooling on a quantum computer''}
\maketitle

{
  \hypersetup{linkcolor=black}
  \tableofcontents
}

\section{Related works} \label{Sec:relatedworks}

In this section, we review the related ground-state preparation and eigenstate preparation works and compare their performances.

Variational quantum eigensolver is a major type of the ground-state preparation method~\cite{peruzzo2014variational,McClean2016theory}. In the variational methods, people try to search the ground state of a Hamiltonian in a pre-determined ansatz, usually characterized by parametric quantum circuits. They then adjust and optimize the parameters in the quantum circuit based on the measured value of energy. 
The quantum approximate optimization algorithm~\cite{farhi2014quantum} is another type of variational methods aiming to solve the combinatorial optimization problems based on an variational evolution combined with adiabatic evolution.
The variational methods are suitable for the near-term quantum computers, since it usually requires no ancillary qubit and shallow circuits. The weakness of the variational methods is that the effectiveness always depends on the choice of the ansatz, whose validity varies for different Hamiltonian problems.

Adiabatic state preparation~\cite{albash2018adiabatic} is an experimentally-fridently non-variational ground-state method based on a time-dependent Hamiltonian evolution. In the adiabatic state preparation, people first prepare the ground state of a simple Hamiltonian $H_0$ and then slowly evolve it under a Hamiltonian that changes slowly from $H_0$ to the target Hamiltonian $H$. Based on the adiabatic theorem, the resulting state is then close the the ground state. Unlike the variational algorithms, the adiabatic algorithms are universal and valid without ansatz assumptions. When it comes to the practical usage, the adiabatic algorithm has two downsides. First, the requried evolution time $t$ depends inverse polynomially to the minimum spectral gap along the entire path from $H_0$ to $H$. Second, the required evolution time $t$ is $\mc{O}(1/\varepsilon)$, where $\varepsilon$ is the infidelity of the target state. 

Phase estimation~\cite{kitaev1995quantum,nielsen2002quantum} is another commonly-used non-variational method which prepare the eigenstate of a given unitary $U$ based on the controlled-$U$ evolution and complementary-basis measurement on the ancillary qubits. In the canonical phase estimation algorithm~\cite{nielsen2002quantum}, in order to prepare the eigenstate of the Hamiltonian $H$, people introduce $\omega$ qubits prepared on the state $\ket{+}$, perform controlled-$e^{-2\pi i H t}$ gate sequentially from one ancillary qubit to the original system with different time $t$, perform inversed Fourier transform on the ancillary qubits and then measure them. The required number of qubit $\omega$ is logarithmic to the inverse of the eigenenergy precision. The circuit depth reflected by the summation of the controlled evolution time $t$, is $\mc{O}(1/\varepsilon)$, where $\varepsilon$ is the infidelity of the eigenstate preparation. The phase estimation can be improved to an iterative version with only one ancillay qubit~\cite{kitaev1995quantum}. However, similar to the adiabatic methods, the required circuit depth of the phase estimation is fundamentally $\mc{O}(1/\varepsilon)$, which is unfavorable if we want to achieve a high-precision preparation.

To further improve the efficiency of the above two nonvariational methods, the linear-combination-of-unitary (LCU) methods have been proposed~\cite{Ge19,lu2021algorithms}. In Ref.~\cite{Ge19}, the authors use $\cos^M(H)$ as a spectral projector to the ground state. They decompose the Hamiltonian function $\cos^M(H)$ based on a LCU formula~\cite{childs2012Hamiltonian,berry2015simulating}, and realize the function $\cos^M(H)$ using $\mc{O}(\log(\frac{1}{\Delta}\log\frac{1}{\varepsilon}))$ ancillar qubits and amplitude amplification. Here, $\Delta$ is the known lower bound of the Hamiltonian energy gap and. Later on, by introducing a block-encoding method~\cite{berry2015Hamiltonian} and a quantum-signal-processing technique~\cite{low2017optimal} to realize a polynomial approximation of the sign function, Lin and Tong~\cite{lin2020nearoptimalground} further reduce the number of required ancillary qubit, making it independent of the precision requirement. The query complexity of these methods are $\mc{O}(\frac{1}{\Delta p_0}\log(\frac{1}{\varepsilon p_0})$, which is exponentially improved comparing to the phase estimation algorithm. The difficulties of realizing these LCU or qubitization-based methods in the near-term is that, they both require many ancillary resources for LCU or block-encoding; furthermore, the realization of the select-$H$ oracles in the LCU or qubitization approach is complicated, especially for the case of non-sparse Hamiltonian $H$ where the number of components in $H$ is large.

In all the above non-variantional ground-state preparation methods, the efficiency relies on the following two assumptions. We remark that, the adiabatic state preparation methods also rely on a function of the state overlap $p_0$ which depends on the adiabatic path connecting $H_0$ and $H$. 

\begin{assumption}[nonvanishing state overlap assumption~\cite{albash2018adiabatic,Ge19,lin2020nearoptimalground}] \label{ass:overlap}
For an $n$-qubit gapped Hamiltonian $H$ with the ground state $\ket{u_0}$, we assume that it is feasible to prepare an initial state $\ket{\psi_0}$ satisfying a nonvanishing overlap of the target eigenstate $p_0 := |\braket{\psi_0|u_0}|^2$. We assume that the lower bound of the overlap satisfies  
$ 
p_0 \geq \Omega (\rm{Poly} (\frac{1}{n}) ).
$
\end{assumption}

\begin{assumption}[Gapped Hamiltonian assumption~\cite{albash2018adiabatic,Ge19,lin2020nearoptimalground}] \label{ass:gap}
For a $n$-qubit gapped Hamiltonian $H$ with the ground state energy $E_0$, we assume that the energy gap $E_1-E_0$ is lower bounded by a known value $\Delta$.
\end{assumption}

Assumption~\ref{ass:overlap} is reasonable in many practical scenarios, since even in many strongly-correlated quantum systems or  chemical molecules Hamiltonians, the mean-field solution, like Hartree-Fock state, still has a considerable overlap with the ground state. In practice, people usually apply the variational methods~\cite{mcardle2020quantum} or adiabatic methods~\cite{oh2008quantum} as a heuristic way to prepare an initial state with a large overlap $p_0$ with the ground state. 
Assumption~\ref{ass:gap} is also of practical relevance for the study of quantum many-body systems. For example, the ferromagnetic XXZ Heisenberg chain $H=\sum_i \sigma_{i}^{x}\sigma_{i+1}^{x}+\sigma_{i}^{y}\sigma_{i+1}^{y}+ \Delta \sigma_{i}^{z}\sigma_{i+1}^{z}$ with $J<0$ has a finite energy gap if  $\Delta > 1$~\cite{koma1997spectral}.

In this work, we propose a quantum algorithm to estimate the ground state energy as well as any ground-state properties with the two assumptions above. Our algorithm has the following advantages,
\begin{enumerate}
\item We only need one ancillary qubit in the whole algorithm. Furthermore, we do not need any quantum oracle, like select-$H$ for LCU or qubitization methods, which could be difficult to be implemented.
\item The time complexity for the ground-state energy estimation reaches the Heisenberg limit.
\item The time complexity for the ground-state property estimation is logarithmic to the inverse of the precision requirement when we adopt the Gaussian cooling function.
\end{enumerate}
To our knowledge this is the first algorithm to achieve all the advantages above. 
In Table~\ref{table:EnergyComp} and Table~\ref{table:ObservableComp}, we compare our method with several typical quantum algorithms in the two primary tasks considered in this work, i.e., energy estimation and observable estimation. Here, we mainly consider the algorithms without the usage of oracles for the block-encoding of a Hamiltonian.
In the Table, we compare the maximal evolution time that characterizes the complexity of quantum circuits that realizes Hamiltonian simulation $e^{-iHt}$, the number of repetitions required to run the quantum circuit, and the total evolution time needed for estimating the ground state energy to within error $\epsilon$. The complexity of estimating properties of excited states  can be analysed in a similar fasion. Note that we simply compare with the conventional phase estimation method.
Here we also use the asymptotic notations besides the usual $\mathcal{O}$ notation 
to denote the complexity up to a polylogorithmic factors, similarly to that in Ref.~\cite{Ge19}.

\begin{table}[h!]
\centering
\begin{tabular}{c c c c c} 
 \hline
\hline
Estimation               & Max time                & Repetition           & Extra qubits        & Total time complexity\\
\hline
This work               & $\mathcal{O}(\mathrm{log}(p_0^{-1})\epsilon^{-1})$ &$\mathcal{O}(p_0^{-2})$   & 1        &$ {\mathcal{O}}(\mathrm{log}(p_0^{-1}) p_0^{-2}\epsilon^{-1})$           \\
Lin \& Tong 21   \cite{lin2021heisenberglimited}         & $\mathcal{O}(\mathrm{log}(p_0^{-1}\epsilon^{-1})\epsilon^{-1})$ &$\mathcal{O}(p_0^{-2})$   & 1        &$ {\mathcal{O}}(\mathrm{log}(p_0^{-1}\epsilon^{-1})p_0^{-2}\epsilon^{-1})$           \\
Time series     \cite{somma2019quantum}      & $\mathcal{O}(\mathrm{Polylog}(p_0^{-1}\epsilon^{-1})\epsilon^{-1})$ &$\mathcal{O}(\epsilon^{-3}p_0^{-2})$ & 1        &$\tilde {\mathcal{O}}(p_0^{-2}\epsilon^{-4})$           \\
Phase estimation              & $\tilde {\mathcal{O}}( p_0^{-1}\epsilon^{-1})$ &$\mathcal{O}(p_0^{-1})$   & $\log(\epsilon^{-1})+\log (p_0^{-1})$     &$\tilde {\mathcal{O}}(p_0^{-2}\epsilon^{-1})$           \\
Projection    \cite{Ge19}          & $\tilde {\mathcal{O}}( p_0^{-1/2} \epsilon^{-{3}/{2}})$ &$\tilde{\mathcal{O}}( p_0^{-1/2})$   & $\log(\epsilon^{-1}) $    &$\tilde {\mathcal{O}}(p_0^{-1}\epsilon^{-3/2})$           \\
\hline
\hline
\end{tabular}
\caption{Comparison of ground state energy estimation. The result of this work  in the table is based on Theorem~\ref{thm:eigenenergy} using Gaussian cooling function. We mainly compare with the methods without the block encoding of Hamiltonian.}
\label{table:EnergyComp}
\end{table}

\begin{table}[h!]
\centering
\begin{tabular}{c c c c c} 
\hline
\hline
Estimation               & Max time                & Repetition           & Extra qubits        & Total time complexity\\
\hline
This work               & $\mathcal{O}(\mathrm{log}(p_0^{-1}\epsilon^{-1})\Delta^{-1})$ &$\mathcal{O}(p_0^{-2}\epsilon^{-2})$   & 1        &$\tilde {\mathcal{O}}(p_0^{-2}\epsilon^{-2}\Delta^{-1})$           \\
Phase estimation  (known $E_0$)         & $\tilde {\mathcal{O}}( p_0^{-1}\epsilon^{-1}\Delta^{-1})$ &$\mathcal{O}(p_0^{-1}\epsilon^{-2})$   & $\log(\epsilon^{-1})+\log (\Delta^{-1} )$    &$\tilde {\mathcal{O}}(p_0^{-2}\epsilon^{-3}\Delta^{-1})$           \\
Projection  (known $E_0$)   \cite{Ge19}       & $\tilde {\mathcal{O}}( p_0^{-1/2} \Delta^{-1})$ &$\tilde {\mathcal{O}}(p_0^{-1/2}\epsilon^{-2})$   & $\log\log(\epsilon^{-1})+\log(\Delta^{-1})$     &$\tilde {\mathcal{O}}(p_0^{-1}\epsilon^{-2}\Delta^{-1})$           \\
Projection  (Unknown $E_0$)   \cite{Ge19}       & $\tilde {\mathcal{O}}( p_0^{-1/2} \Delta^{-{3}/{2}})$ &$\tilde {\mathcal{O}}(p_0^{-1/2}\epsilon^{-2}\Delta^{-{1}/{2}})$   & $\log\log(\epsilon^{-1})+\log(\Delta^{-1})$     &$\tilde {\mathcal{O}}(p_0^{-1}\epsilon^{-2}\Delta^{-{2}})$           \\
\hline
\hline
\end{tabular}
\caption{Comparison of observable estimation on the ground state. The result of this work in the table is based on Theorem~\ref{thm:observable} using Gaussian cooling function. We mainly compare with the methods without the block encoding of Hamiltonian.
}
\label{table:ObservableComp}
\end{table}

We remark that, one can generalize Assumption~\ref{ass:overlap} and Assumption~\ref{ass:gap} to the case of the estimation of $j$-th eigenenergy $E_j$ and the properties of $j$-th eigenstate $\ket{u_j}$: we assume the initial state has a large overlap $p_j$ with the $j$-th eigenstate of the target Hamiltonian $H$; we also assume a known lower bound $\Delta$ on the energy gap $\min\{|E_j-E_{j-1}|, |E_{j}-E_{j+1}|\}$. In this case, our proposed algorithm can then be used for eigenenergy and eigenstate property estimation.

\section{Cooling functions and the dual phase representations} \label{Sec:cooling_func}

In this work, we consider an $n$-qubit system with a gapped Hamiltonian $H$. The eigenstate $\ket{u_i}$ and the corresponding eigenenergy $E_i$ of the Hamiltonian satisfy,
\begin{equation}
    H \ket{u_i} = E_i \ket{u_i}, \quad i=0,1,...,N-1.
\end{equation}
Here, $N:= 2^n$. 

We define the cooling function $g(h)$ as follows.
\begin{definition} \label{Def:coolingfunc}
A real-valued function $g(h):\mbb{R}\to\mbb{R}$ is called a cooling function if it satisfies,
\begin{enumerate}
    \item The value of $g(0)$ is non-zero: $g(0)\neq 0$;
    \item The absolute value of $g(h)$ is a single peak shape: $g(h') \leq g(h)$, $\forall h'>h>0$ and $\forall h'<h<0$;
    \item Th asymptotic value of $h\neq 0$ is vanishing: $\lim_{\tau\to\infty} \left|\frac{g(\tau h)}{g(0)}\right|=0$, $\forall h\neq 0$.
\end{enumerate}
\end{definition}

For a given $n$-qubit Hamiltonian, suppose we want to prepare the $j$-th eigenstate $\ket{u_j}$ with eigenenergy $E_j$. We define a cooling operator as follows,
\begin{equation} \label{eq:g_HEj}
    g( \tau(H-E_j) ) := \sum_{i=0}^{N-1} g(\tau (E_i - E_j)) \ket{u_i}\bra{u_i}.
\end{equation}
Then, for any given initial state $\ket{\psi_0}$ with $|\braket{\psi_0|u_j}|^2 \neq 0$, the cooling operator evolve $\ket{\psi_0}$ to the eigenstate $\ket{u_j}$,
\begin{equation}
    \lim_{\tau\to\infty} g( \tau(H-E_j) ) \ket{\psi_0} \propto \ket{u_j}.
\end{equation}

The cooling operator is usually a nonphysical operation. To realize it, we consider its dual realization based on a Fourier transform,
\begin{equation} \label{eq:FourierTrans}
\begin{aligned}
    f(x) &= \int_{-\infty}^{\infty} g(h) \, e^{-ixh} dh, \\
    g(h) &= \frac{1}{2\pi} \int_{-\infty}^{\infty} f(x) \, e^{ixh} dx.
\end{aligned}
\end{equation}

We define the norm of the dual function $f(x)$ as
\begin{equation}
        \|f\| := \int_{-\infty}^{\infty} |f(x)| dx.
\end{equation}
Suppose $\|f\|$ is finite, we can decompose $f(x)$ to a normalized probability distribution with an extra phase,
\begin{equation} \label{eq:fxtopx}
    f(x) = |f(x)| e^{i\phi(x)} = \|f\| \, p(x) e^{i\phi(x)}.
\end{equation}
The Fourier transform in Eq.~\eqref{eq:FourierTrans} can then be re-written as
\begin{equation} \label{eq:ghpx}
\begin{aligned}
    p(x) &= \frac{1}{\|f\|} \int_{-\infty}^{\infty} g(h) \, e^{-i\phi(x)} e^{-ixh} dh, \\
    g(h) &=  \frac{\|f\|}{2\pi} \int_{-\infty}^{\infty} p(x) \, e^{i\phi(x)} e^{ixh} dx.
\end{aligned}
\end{equation}

Based on Eq.~\eqref{eq:ghpx}, the cooling operator $g(\tau(H-E_j))$ in Eq.~\eqref{eq:g_HEj} can be expanded as,
\begin{equation} \label{eq:coolingop_exp}
\begin{aligned}
    g(\tau (H-E_j)) &= \frac{\|f\|}{2\pi} \int_{-\infty}^{\infty} p(x) e^{i\phi(x)} e^{ix \tau(H-E_j)} dx, \\
    &= \frac{\|f\|}{2\pi} \int_{-\infty}^{\infty} p(x) e^{i(\phi(x) - \tau x E_j)} e^{i \tau x H} dx.
\end{aligned}
\end{equation}
Therefore, to realize a quantum cooling operator $g(\tau(H-E_j))$, we can expand it to a weighted superposition of the unitary operators $e^{itH}$ with different evolution time $t=\tau x$ ranging from $-\infty$ to $\infty$. 

In practice, the evolution time is a limited number related to the depth of the quantum circuit. We define the truncated cooling function as,
\begin{equation}
    \tilde{g}(h;x_m) = \frac{\|f\|}{2\pi} \int_{-x_m}^{x_m} p(x) \, e^{i\phi(x)} e^{ixh} dx.
\end{equation}

The difference of the $\tilde{g}(h;x_m)$ from $g(h)$ is bounded by
\begin{equation}
\begin{aligned}
    |g(h) - \tilde{g}(h;x_m)| &= \frac{\|f\|}{2\pi} \left| \int_{-\infty}^{-x_m} p(x) \, e^{i\phi(x)} e^{ixh} dx + \int_{x_m}^{\infty} p(x) \, e^{i\phi(x)} e^{ixh} dx \right| \\
    &\leq \frac{\|f\|}{2\pi} \int_{-\infty}^{-x_m} | p(x) \, e^{i\phi(x)} e^{ixh}| dx + \int_{x_m}^{\infty} |p(x) \, e^{i\phi(x)} e^{ixh}| dx \\
    &= \frac{\|f\|}{2\pi} \left( 1 - \int_{-x_m}^{x_m} p(x)dx \right).
\end{aligned}
\end{equation}
Later in Sec.~\ref{Sec:cooling_error_complexity} we will show that, the estimation error caused by the finite evolution time will be bounded by the tail probability of $p(x)$. 

\begin{definition} \label{Def:realizable}
We say a cooling function $g(h)$ in Definition~\ref{Def:coolingfunc} is \textit{realizable} if the following requirements hold. 
\begin{enumerate}[label=C\arabic*.,start=0]
    \item The Fourier transform of $g(h)$ exists,
    \begin{equation}
    f(x) = \int_{-\infty}^{\infty} g(h) e^{-ixh} dh.
    \end{equation}
    \item The norm of the dual function $\|f\|$
    \begin{equation}
        \|f\| := \int_{-\infty}^{\infty} |f(x)| dx,
    \end{equation}
    is finite. 
    \item $\big|1-\int^{L(\varepsilon)}_{-L(\varepsilon)} p(x) dx \big|\le \varepsilon, ~\forall \varepsilon\ge 0,~ \exists L(\varepsilon)=\mathcal{O}(\mathrm{poly}(\frac{1}{\varepsilon}))$.
\end{enumerate}
\end{definition}

Now, we introduce several typical cooling functions $g(h)$ and dicuss their properties. We will focus on the cutoff cost $L(\varepsilon)$ defined in Definition~\ref{Def:realizable} which shows the hardness of realizing $g(h)$ with real-time sampling.


\subsection{Rectangular function}

The rectangular function $\mr{rect}(h)$ is defined as,
\begin{equation}
\mr{rect}(h) = 
\begin{cases}
1, & |h|\leq \frac{1}{2}, \\
0, & |h| > \frac{1}{2}.
\end{cases}
\end{equation}
It satisties Definition~\ref{Def:coolingfunc} and hence a cooling function. When acting as a cooling function $g(\tau(H-E_j))$, the rectangular function is an ideal ``energy band filter'', as it projects the state to the energy subspace in the range $[E_j-\frac{1}{2\tau}, E_j+\frac{1}{2\tau}]$. 

The Fourier transform of $\mr{rect}(h)$ is the Sinc function,
\begin{equation}
    f(x) = \sinc\left(\frac{x}{2\pi}\right),
\end{equation}
where $\sinc(x):= \frac{\sin x}{x}$.

Unfortunately, the norm of $f(x)$ is infinite,
\begin{equation}
\begin{aligned}
\|f\| &= \int_{-\infty}^{\infty} |\sinc\left( \frac{x}{2\pi} \right)| \\
&= 4\pi \int_0^{\infty} |\sinc y| dy \\
&= 4\pi \sum_{k=0}^{\infty} \int_{k\pi}^{(k+1)\pi} |\frac{\sin y}{y}| dy \\
&> 4\pi \sum_{k=0}^{\infty} \int_{k\pi}^{(k+1)\pi} \frac{|\sin y|}{(k+1)\pi} dy \\
&= 4\pi \sum_{k=0}^{\infty} \frac{1}{k+1}.
\end{aligned}
\end{equation}
From the divergence of $\sum_{k=0}^{\infty} \frac{1}{k+1}$ we know that $\|f\|=\infty$. As a result, the rectangular function is not a realizable cooling function based on Definition~\ref{Def:realizable}.

\subsection{Triangular function}

The triangular function $\mr{tri}(h)$ is defined as,
\begin{equation}
\mr{tri}(h) = 
\begin{cases}
1 - |h|, & |h|\leq 1, \\
0, & |h| > 1.
\end{cases}
\end{equation}
It satisties Definition~\ref{Def:coolingfunc} and hence a cooling function. Similar to the rectangular function, it can be used as an ``energy band filter'', since it filters the energy in the range $[E_j-\frac{1}{\tau}, E_j+\frac{1}{\tau}]$. However, it will modify the weight of different energy values in that range. When $h\geq 0$, the inverse function of $\mr{tri}(h)$ is
\begin{equation} \label{eq:tri_inv}
\mr{tri}^{-1}(p) = 1 - p \leq 1.
\end{equation}

The Fourier transform of $\mr{rect}(h)$ is the square of Sinc function,
\begin{equation}
    f(x) = \sinc^2 \left(\frac{x}{2\pi}\right).
\end{equation}

The norm of $f(x)$ is finite,
\begin{equation}
\|f\| = 2\pi \int_{-\infty}^{\infty} \sinc^2(y) dy = 2\pi^2.
\end{equation}
The sample probability distribution is then,
\begin{equation} \label{eq:px_tri}
p(x) = \frac{1}{\|f\|}f(x) = \frac{1}{2\pi^2}\sinc^2 \left(\frac{x}{2\pi}\right).
\end{equation}

To check the realizability of cooling function $g(h)$, we need to study the property of the tail probability of $p(x)$ with respect to the cutoff sample time $x_m$. Solving the following equation,
\begin{equation}
\varepsilon = \int_{-x_m}^{x_m} p(x) dx,
\end{equation}
where $p(x)$ is defined in Eq.~\eqref{eq:px_tri}, we have
\begin{equation}
\varepsilon = \frac{2}{\pi}\left( \frac{\pi}{2}- \mr{Si}(\frac{x_m}{\pi})  \right) + \frac{2 - 2\cos(\frac{x_m}{\pi}) }{x_m},
\end{equation}
where $\mr{Si}(x):= \int_0^{x} \sinc(y) dy$ is the Sinc integral function. 

\begin{proposition} 
We have
\begin{equation}
\frac{\pi}{2} - \mr{Si}(x) \leq \frac{1}{x}.
\end{equation}
\end{proposition}

\begin{proof}
By definition, we have,
\begin{equation}
\frac{\pi}{2} - \mr{Si}(x) = \int_x^{\infty} \frac{\sin(t)}{t} dt  = \int_x^{\infty} \frac{\cos x \sin(t-x) + \sin x \cos(t-x)}{t} dt.
\end{equation}
Now we denote $f(t) = \sin x\cos t + \cos x\sin t$, then
\begin{equation}
\begin{aligned}
\frac{\pi}{2} - \mr{Si}(x) = \int_x^{\infty} \frac{f(t-x)}{t} dt.
\end{aligned}
\end{equation}
We introduce the Laplace transform of $f(t)$,
\begin{equation}
\mc{L}[f(t)](s) = \int_0^{\infty} f(t) e^{-st} dt,
\end{equation}
then,
\begin{equation} \label{eq:Lfu}
\begin{aligned}
\frac{\pi}{2} - \mr{Si}(x) = \mc{L}\left[ \frac{f(t-x) u(t-x)}{t} \right](0),
\end{aligned}
\end{equation}
where $u(t):= \eta(t\geq 0)$.

From the property of Laplace transform, we have
\begin{equation} \label{eq:Lfexp1}
\begin{aligned}
\int_{s}^\infty \mc{L}[f(t)](p)\,e^{-x p} dp &= \int_{s}^\infty \mc{L}[f(t-x)u(t-x)] dp \\
&= \mc{L}\left[\frac{f(t-x)u(t-x)}{t}\right](s).
\end{aligned}
\end{equation}
On the other hand, we know that 
\begin{equation}
\mc{L}[f(t)](p) = \mc{L}[\sin x\cos t + \cos x\sin t](p) = \frac{\sin x\, p + \cos x}{p^2+1},
\end{equation}
then
\begin{equation} \label{eq:Lfexp2}
\begin{aligned}
\int_{s}^\infty \mc{L}[f(t)](p)\,e^{-x p} dp &= \int_{s}^\infty \frac{\sin x\, p + \cos x}{p^2+1} \,e^{-x p} dp.
\end{aligned}
\end{equation}
Taking $s=0$ and combine Eq.~\eqref{eq:Lfu},~\eqref{eq:Lfexp1},~and \eqref{eq:Lfexp2}, we have
\begin{equation}
\frac{\pi}{2} - \mr{Si}(x) = \int_0^{\infty} \frac{\sin x\, p + \cos x}{p^2+1} \,e^{-x p} dp.
\end{equation}
Using Cauchy-Schwarz inequality we can prove that,
\begin{equation}
p\sin x + 1\cdot \cos x \leq \sqrt{1+p^2}\sqrt{\sin^2 x + \cos^2 x} = \sqrt{1+p^2},
\end{equation}
hence
\begin{equation}
\begin{aligned}
\frac{\pi}{2} - \mr{Si}(x) &= \int_0^{\infty} \frac{\sin x\, p + \cos x}{p^2+1} \,e^{-x p} dp \\
&\leq \int_0^{\infty} \frac{1}{\sqrt{p^2+1}} \,e^{-x p} dp \\
&\leq \int_0^{\infty} e^{-x p} dp \\
&= \frac{1}{x}.
\end{aligned}
\end{equation}
\end{proof}

Therefore,
\begin{equation}
\begin{aligned}
\varepsilon &\leq \frac{2}{\pi} \frac{\pi}{x_m} + \frac{2-2\cos(\frac{x_m}{\pi})}{x_m} \\
&\leq \frac{2}{x_m} + \frac{4}{x_m} = \frac{6}{x_m}.
\end{aligned}
\end{equation}
We have
\begin{equation}
x_m \leq \frac{6}{\varepsilon} =: L(\varepsilon) = \mathcal O(\mathrm{poly}(\frac{1}{\varepsilon})).
\end{equation}

Therefore, the triangle function is a realizable cooling function. However, the circuit depth requirement of triangle function is still demanding, as the tail probability $\varepsilon$ decays slowly with the cutoff time $x_m$.

\subsection{Double-side exponential decay function}

The double-side exponential decay function is 
\begin{equation}
g(h):= e^{-|h|}.
\end{equation} 
Hereafter we call it exponential function for simplicity. It satisties Definition~\ref{Def:coolingfunc} and hence a cooling function.
The exponential function is of special importance as it describes the imaginary time evolution $e^{-\tau H}$ when the eigenenergies are all positive. When $h\geq 0$, the inverse function is
\begin{equation} \label{eq:exp_inv}
g^{-1}(p) = \ln(\frac{1}{p}).
\end{equation}

The Fourier transform of the exponential function is 
\begin{equation}
f(x) = \frac{2}{1+x^2}.
\end{equation}
The norm of $f(x)$ is finite,
\begin{equation}
\|f\| = \int_{-\infty}^{\infty} \frac{2}{1+x^2} dx = 2\pi.
\end{equation}

The sample probability distribution is,
\begin{equation} \label{eq:px_exp}
p(x) = \frac{1}{\|f\|}f(x) = \frac{1}{2\pi} \frac{2}{1+x^2},
\end{equation}
which is a Lorentian distribution.

Now, we check the tail probability of $p(x)$ with respect to the cutoff sample time $x_m$. Solving the following equation,
\begin{equation}
\varepsilon = \int_{-x_m}^{x_m} p(x) dx,
\end{equation}
where $p(x)$ is defined in Eq.~\eqref{eq:px_exp}, we have
\begin{equation}
x_m = \tan\left[\frac{\pi}{2}(1-\varepsilon)\right],
\end{equation}

To calculate the dependence of $x_m$ with respect to $\frac{1}{\varepsilon}$, we expand $\frac{1}{x_m}$ at the point $\varepsilon \to 0^+$ using Taylor expansion with reminder,
\begin{equation} \label{eq:tanexpand}
\begin{aligned}
\frac{1}{x_m} &= \frac{1}{\tan\left[\frac{\pi}{2}(1-\varepsilon)\right]} \\
&= 0 + \frac{\pi}{2} \frac{1}{\cos^2(\frac{\pi}{2}\cdot0)} \varepsilon + \frac{\pi^2}{4}\frac{\sin(\pi\xi)}{\cos^4(\frac{\pi}{2}\xi)} \frac{\varepsilon^2}{2} \\
&>  \frac{\pi}{2} \varepsilon,
\end{aligned}
\end{equation}
where $\xi\in(0,\varepsilon)$. We then have,
\begin{equation}
x_m < \frac{2}{\pi}\frac{1}{\varepsilon} =: L(\varepsilon) = \mathcal O(\mathrm{poly}(\frac{1}{\varepsilon})).
\end{equation}
Therefore, the exponential function is a realizable cooling function. Still, the circuit depth requirement is demanding, as the tail probability $\varepsilon$ decays slowly with the cutoff time $x_m$. In what follows, we consider two cooling functions with small tails.

\subsection{Gaussian function}

The Gaussian function we used here is 
\begin{equation}
g(h) := e^{-h^2}.
\end{equation}
It satisfies Definition~\ref{Def:coolingfunc} and hence a cooling function. When $h\geq 0$, the inverse function is
\begin{equation} \label{eq:gau_inv}
g^{-1}(p) = \sqrt{\ln(\frac{1}{p})}.
\end{equation}

The Fourier transform of Gaussian function is still a Gaussian function,
\begin{equation}
f(x) = \sqrt{\pi} e^{-\frac{x^2}{4}}.
\end{equation}
The norm of $f(x)$ is finite,
\begin{equation}
\|f\| = \sqrt{\pi} \int_{-\infty}^{\infty} e^{-\frac{x^2}{4}} dx = 2\pi.
\end{equation}

The sample probability distribution is,
\begin{equation} \label{eq:px_gau}
p(x) = \frac{1}{\|f\|}f(x) = \frac{1}{2\pi} \sqrt{\pi} e^{-\frac{x^2}{4}},
\end{equation}
which is a Gaussian distribution.

Now, we check the tail probability of $p(x)$ with respect to the cutoff sample time $x_m$. Solving the following equation,
\begin{equation}
\varepsilon = \int_{-x_m}^{x_m} p(x) dx,
\end{equation}
where $p(x)$ is defined in Eq.~\eqref{eq:px_gau}, we have
\begin{equation}
\varepsilon = \mr{erfc}(\frac{x_m}{2}),
\end{equation}
where $\mr{erfc}(x):= 1 - \frac{2}{\sqrt{\pi}}\int_0^x e^{-t^2} dt$ is the complementary of the error function. It can be upper-bounded by a Chernoff-type formula~\cite{karagiannidis2007improved},
\begin{equation}
\mr{erfc}(x) \leq \frac{1}{\sqrt{\pi}x} e^{-x^2}, \quad x\geq 0.
\end{equation}
Therefore, 
\begin{equation}
\begin{aligned}
& \varepsilon = \mr{erfc}(\frac{x_m}{2}) \leq \frac{2}{\sqrt{\pi}x_m} e^{-\frac{x_m^2}{4}} \\
\Rightarrow\quad & \ln \frac{\sqrt{\pi}\varepsilon}{2} \leq -\frac{x_m^2}{4} - \ln x_m \\
\Rightarrow\quad & \frac{x_m^2}{4} \leq \frac{x_m^2}{4} + \ln x_m \leq -\ln \left(\frac{\sqrt{\pi}\varepsilon}{2} \right) \\
\Rightarrow\quad & x_m \leq 2 \sqrt{-\ln\varepsilon - \ln\frac{\sqrt{\pi}}{2}} \leq 2\sqrt{\ln\left(\frac{1}{\varepsilon}\right)}.
\end{aligned}
\end{equation}

We have $L(\varepsilon) = 2\sqrt{\ln\left(\frac{1}{\varepsilon}\right)} = \mathcal O(\mathrm{poly}(\frac{1}{\varepsilon}))$. Therefore, the Gaussian function is a realizable cooling function. Note that, unlike the triangular function or the exponential function, the tail of Gaussian function decays quickly with respect to $x_m$. This implies that the Gaussian cooling is more experiementally friendly.

\subsection{Secant hyperbolic function}

The secant hyperbolic function is 
\begin{equation}
\sech(h):= \frac{2}{e^h + e^{-h}}. 
\end{equation}
It satisfies Definition~\ref{Def:coolingfunc} and hence a cooling function. The inverse of $\sech$ is hard to be presented in an analytical form, but we can easily derive an upper bound for it,
\begin{equation} \label{eq:sech_inv}
g^{-1}(p) < \ln\frac{2}{p}.
\end{equation}

The Fourier transform of secant hyperbolic function is still a secant hyperbolic function,
\begin{equation}
f(x) = \pi\, \sech\left(\frac{\pi}{2}x\right).
\end{equation}
The norm of $f(x)$ is finite,
\begin{equation}
\|f\| = \pi \int_{-\infty}^{\infty} \sech\left(\frac{\pi}{2}x\right) dx = 2\pi.
\end{equation}

The sample probability distribution is,
\begin{equation} \label{eq:px_sech}
p(x) = \frac{1}{\|f\|}f(x) = \frac{1}{2\pi} \pi\, \sech\left(\frac{\pi}{2}x\right),
\end{equation}
which is a secant hyperbolic distribution.

Now, we check the tail probability of $p(x)$ with respect to the cutoff sample time $x_m$. Solving the following equation,
\begin{equation}
\varepsilon = \int_{-x_m}^{x_m} p(x) dx,
\end{equation}
where $p(x)$ is defined in Eq.~\eqref{eq:px_sech}, we have
\begin{equation} 
\begin{aligned}
& x_m = \frac{2}{\pi} \ln\left\{\tan \left[ \frac{\pi}{2}(1-\frac{\varepsilon}{2}) \right] \right\} \\
\Rightarrow \quad & e^{ -\frac{\pi}{2}x_m } = \frac{1}{\tan\left[\frac{\pi}{2}(1-\frac{\varepsilon}{2})\right]}.
\end{aligned}
\end{equation}
Similar to Eq.~\eqref{eq:tanexpand}, we expand the inversed tangent function at the point $\varepsilon\to 0^+$ using Taylor expansion with reminder,
\begin{equation} \label{eq:tanexpand2}
\begin{aligned}
e^{ -\frac{\pi}{2}x_m } &= \frac{1}{\tan\left[\frac{\pi}{2}(1-\frac{\varepsilon}{2})\right]} \\
&= 0 + \frac{\pi}{4} \frac{1}{\cos^2(\frac{\pi}{4}\cdot0)} \varepsilon + \pi^2 \frac{ \sin^4(\frac{\pi}{4}\xi) }{ \sin^3(\frac{\pi}{2}\xi) } \frac{\varepsilon^2}{2} \\
&>  \frac{\pi}{4} \varepsilon,
\end{aligned}
\end{equation}
where $\xi\in(0,\varepsilon)$.
Therefore,
\begin{equation} \label{eq:tanexpand3}
\begin{aligned}
x_m < \frac{2}{\pi} \ln(\frac{4}{\pi}\frac{1}{\varepsilon})=: L(\varepsilon) = \mathcal O(\mathrm{poly}(\frac{1}{\varepsilon})).
\end{aligned}
\end{equation}
Therefore, the secant hyperbolic function is also a realizable function. The overhead of secant hyperbolic function is better then the Gaussian function; however, the asymptotic cost of the Gaussian function is quadratically better then the secant hyperbolic function.

\section{Detailed quantum cooling algorithm} \label{Sec:cooling_algorithm}

Suppose we want to estimate the expectation value on an eigenstate $\ket{u_j}$, $\braket{O}:=\braket{u_j|O|u_j}$, using the quantum cooling algorithm. For a given finite evolution time $\tau$, the estimation value of $\braket{O}_{\tau}$ is given by
\begin{equation}
\braket{O}_{\tau} = \frac{N_\tau(O)}{D_\tau},
\end{equation}
where
\begin{equation}
\begin{aligned}
D_\tau &= \braket{\psi_0| g^2(\tau(H-E_j)) |\psi_0}, \\
N_\tau(O) &= \braket{\psi_0| g(\tau(H-E_j)) O g(\tau(H-E_j)) |\psi_0}.
\end{aligned}
\end{equation}
are, respectively, the normalization factor and the unnormalized expectation value. To apply the Hadamard test, we decompose the observable $O$ by the Pauli operators, 
\begin{equation} \label{eq:Odecom}
O = \sum_{l\in\mc{P}_n} o_l P_l = \|O\|_1 \sum_{l\in\mc{P}_n} \mr{Pr}_O(l) P_l
\end{equation}
where $\mc{P}_n$ denotes the $n$-qubit Pauli group. $o_l$ is the coefficient of the Pauli component $P_l$. We remark that, the coefficients $\{o_l\}$ are all set to be positive. The signs of the coefficients are put into the corresponding Pauli matrices $\{P_l\}$. In Eq.~\eqref{eq:Odecom}, $\|O\|_1$ is the $l_1$-norm of the Pauli coefficients of $O$,
\begin{equation}
\|O\|_1 = \sum_{l\in\mc{P}_n} o_l.
\end{equation}
The probability distribution $\mr{Pr}_O(l)$ is defined to be,
\begin{equation}
\mr{Pr}_O(l) = \frac{o_l}{\sum_{l\in\mc{P}_n} o_l} = \frac{1}{\|O\|_1} o_l.
\end{equation} 

Using Eq.~\eqref{eq:coolingop_exp}, we expand the cooling operator $g(\tau(H-E_j))$ and receive the following estimation formulas,
\begin{equation} \label{eq:DNestimation}
\begin{aligned}
D_\tau &= \left(\frac{\|f\|}{2\pi}\right)^2  \int_{-\infty}^{\infty} dy \tilde{p}(y) e^{-i\tau y E_j} \braket{\psi_0| e^{i\tau y H} |\psi_0}, \\
N_\tau(O) &= \|O\|_1 \left(\frac{\|f\|}{2\pi}\right)^2 \int_{-\infty}^{\infty} dx \int_{-\infty}^{\infty} dx' p(x) p(x') \sum_{l\in\mc{P}_n} \mr{Pr}_O(l) e^{-i\tau (x-x') E_j} \braket{\psi_0| e^{-i\tau x' H} P_l e^{i\tau x' H} |\psi_0},
\end{aligned}
\end{equation}
where $\tilde{p}(y):= \frac{1}{2}\int_{-\infty}^{\infty}p(\frac{z+y}{2})p(\frac{z-y}{2})dz$ is the self-correlation of the probability $p(x)$. We remark that, the normalization factor $\left(\frac{\|f\|}{2\pi}\right)^2$ can be removed, since it is the same for both $D_\tau$ and $N_\tau(O)$, hence independent of the estimation of $\braket{O}_{\tau}$. In the following discussion, we ignore this normalization factor during the estimation procedure.

In the quantum cooling algorithm, we generate the (normalized) evolution time $x$ (or $y$) based on the sample probability $p(x)$ (or $\tilde{p}(y)$) and sample the innerproduct values using quantum experiments. For a practical consideration, when the (normalized) evolution time $x$ (or $y$) is larger then a cutoff value $x_m$, we don't perform the quantum experiment and direct denote the estimation of this round to be $0$. In this way, the estimation formula of $N_\tau(O)$ and $D_\tau$, originally given by Eq.~\eqref{eq:DNestimation}, now becomes
\begin{equation} \label{eq:DNestimationCut}
\begin{aligned}
D_\tau^{(x_m)} &= \int_{-x_m}^{x_m} dy \tilde{p}(y) e^{-i\tau y E_j} \braket{\psi_0| e^{i\tau y H} |\psi_0}, \\
N_\tau^{(x_m)}(O) &= \|O\|_1 \int_{-x_m}^{x_m} dx \int_{-x_m}^{x_m} dx' p(x) p(x') \sum_{l\in\mc{P}_n} \mr{Pr}_O(l) e^{-i\tau (x-x') E_j} \braket{\psi_0| e^{-i\tau x' H} P_l e^{i\tau x' H} |\psi_0}.
\end{aligned}
\end{equation}

To estimate the values of $N_\tau^{(x_m)}(O)$ and $D_\tau^{(x_m)}$ in Eq.~\eqref{eq:DNestimationCut}, the core issue is to realize the unbiased estimation of the following quantities,
\begin{equation}
\begin{aligned}
& \braket{\psi_0| e^{i\tau y H} |\psi_0}, \\
& \braket{\psi_0| e^{-i\tau x' H} P_l e^{i\tau x' H} |\psi_0}.
\end{aligned}
\end{equation}
These $\braket{\psi|U|\psi}$ form quantities can be estimated using the Hadamard test, shown in Fig.~\ref{fig:HadamardTest}. 
Here, the $S$ gate is
\begin{equation}
S=
\begin{pmatrix}
1 & 0 \\
0 & -i
\end{pmatrix}.
\end{equation}
To measure $\braket{\psi|U|\psi}$, we first prepare the state $\ket{\psi}$ and an extra ancillary qubit prepared on $\ket{+}$. Afterward, we perform a $C$-$U$ gate from ancillary to $\ket{\psi}$. If we directly measure the ancillary qubit on the $X$-basis, the outcome $a$ will be $0$ with a probability of $\frac{1}{2}(1 + \mr{Re}(\braket{\psi|U|\psi}))$ and $1$ with a probability of $\frac{1}{2}(1 - \mr{Re}(\braket{\psi|U|\psi}))$. Alternatively, if we perform an extra $W$ gate before the $X$-basis measurement, the outcome $a$ will be $0$ with a probability of $\frac{1}{2}(1 + \mr{Im}(\braket{\psi|U|\psi}))$ and $1$ with a probability of $\frac{1}{2}(1 - \mr{Im}(\braket{\psi|U|\psi}))$. 

\begin{figure}[htbp]
\centering
\includegraphics[width=6cm]{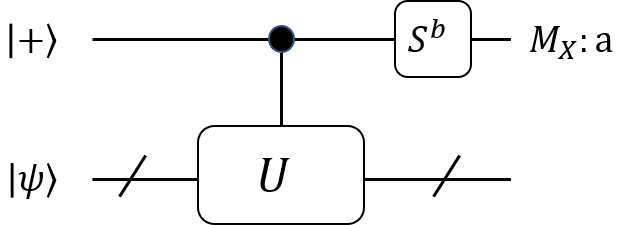}
\caption{The diagram of the Hadamard test. If $b=0$, $\Pr(a=0)=\frac{1}{2}(1 + \mr{Re}(\braket{\psi|U|\psi}))$; if $b=1$, $\Pr(a=0)=\frac{1}{2}(1 + \mr{Im}(\braket{\psi|U|\psi}))$.} \label{fig:HadamardTest}
\end{figure}

To simplify the theoretical analysis, we now introduce a combined estimation of the real part and imaginary part in a single round. In each round of experiment, we first randomly decide the binary value $b$ uniformly from $\{0,1\}$ in Fig.~\ref{fig:HadamardTest}. We denote the binary uniform distribution of $b$ as $\Pr_U(b)$. Based on the measurement outcome $a$, we construct the following complex-valued estimator,
\begin{equation} \label{eq:hatr}
\hat{r} =
\begin{cases}
2, & \quad (b,a)=(0,0),\\
-2, & \quad (b,a)=(0,1),\\
2i, & \quad (b,a)=(1,0),\\
-2i, & \quad (b,a)=(1,1).
\end{cases}
\end{equation}
That is, $\hat{r}=2i^b(-1)^a$. In this case, we have
\begin{equation}
\begin{aligned}
\mbb{E}(\hat{r}) &= \sum_{b\in{0,1}} \Pr_U(b) \Pr(a|b)\, \hat{r}(b,a) \\
&= \frac{1}{4}(1 + \mr{Re}(\braket{\psi|U|\psi}))\cdot 2 + \frac{1}{4} (1 - \mr{Re}(\braket{\psi|U|\psi}))\cdot (-2) \\
&\quad + \frac{1}{4}(1 + \mr{Im}(\braket{\psi|U|\psi}))\cdot (2i) + \frac{1}{4}(1 - \mr{Im}(\braket{\psi|U|\psi}))\cdot (-2i) \\
&= \mr{Re}(\braket{\psi|U|\psi}) + i \mr{Im}(\braket{\psi|U|\psi}) \\
&= \braket{\psi|U|\psi}.
\end{aligned}
\end{equation}
Therefore, $\hat{r}$ is an unbiased estimator for $\braket{\psi|U|\psi}$.

In the estimation of $D_\tau$ and $N_\tau(O)$, we also put the phase term and normalization factor $\|O\|_1$ in Eq.~\eqref{eq:DNestimation} into the estimator. That is, we define
\begin{equation}
\begin{aligned}
\hat{d} &= e^{-i\tau y E_j} \hat{r}, \\
\hat{n} &= \|O\|_1 e^{-i\tau(x-x') E_j} \hat{r}.
\end{aligned}
\end{equation}
As a result, if we randomly sample $y$ from $\tilde{p}(y)$, based on Eq.~\eqref{eq:DNestimation} we have
\begin{equation}
\mbb{E}_{y,b,a} (\hat{d}) = D_\tau.
\end{equation}
On the other hand, if we randomly sample $x,x'$ from $p(x)$ and $l$ from $\mr{Pr}_O(l)$, based on Eq.~\eqref{eq:DNestimation} we have
\begin{equation}
\mbb{E}_{x,x',l,b,a} (\hat{n}) = N_\tau(O).
\end{equation}

In the experiments, we generate random variables $x, x', b$ and $l$ independently in each round of experiment. After performing the single-shot Hadamard test experiements, we will obtain a group of unbiased estimators $\{\hat{d}_p\}_{p=1}^{N_M}$ and $\{\hat{n}_q\}_{p=1}^{N_M}$ for $D_\tau$ and $N_\tau(O)$, respectively. We then use the average value of the estimators as an accurate estimate of $D_\tau$ and $N_\tau(O)$, whose tightness is given by the concentration bound analyzed in Sec.~\ref{Ssc:energyFiniteNM} and Sec.~\ref{Ssc:finiteNM}.

In the main text, we consider two applications of the universal cooling: 1) to estimate the ground state properties when the ground state energy $E_0$ is known; 2) to estimate the properties of the $j$-th eigenstate without the exact knowledge of the energy. We can divide it to three elementary tasks:
\begin{enumerate}
    \item Estimate the eigenenergy $E_j$ of the $j$-th eigenstate given a initial guess interval $[E_j^L, E_j^U]$;
    \item Given a known eigenenergy $E_j$, estimate the normalization factor $D_\tau$;
    \item Given a known eigenenergy $E_j$ and an observable $O$, estimate the normalization factor $N_\tau(O)$.
\end{enumerate}

We now introduce the detailed algorithm for these tasks.

\subsection{Estimate the normalization factor}

We first introduce the algorithm to estimate the normalization factor $D_\tau$ based on Eq.~\eqref{eq:DNestimationCut} and the Hadamard test, when the eigenenergy $E_j$ is known. Recall that the unbiased estimator for the normalization factor is,
\begin{equation} \label{eq:hatdunbiased}
\begin{aligned}
\hat{d} &= e^{-i\tau y E_j} \hat{r} = 2i^b(-1)^a e^{-i\tau y E_j}, \\
\end{aligned}
\end{equation}
If we randomly sample $y$ from $\tilde{p}(y)$, based on Eq.~\eqref{eq:DNestimation} we have
\begin{equation}
\mbb{E}_{y,b,a} (\hat{d}) = D_\tau.
\end{equation}

Algorithm~\ref{Alg:NormEst} below is based on constructing the estimator in Eq.~\eqref{eq:hatdunbiased}. We also introduce the truncation into the sampling of $y$.

\begin{figure}
\begin{algorithm}[H]
        \caption{Normalization factor estimation}
        \label{Alg:NormEst}
        \begin{algorithmic}[1]
            \Require
            An $n$-qubit Hamiltonian $H$; initial state $\ket{\psi_0}$ with nonzero overlap with $j$-th eigenstate of $H$: $p_j = |\braket{u_j|\psi_0}|^2$; the energy $E_j$ for the $j$-th eigenstate; cooling function $g(h)$ and the corresponding sampling probability $p(x)$; proper choice of the imaginary time $\tau$, normalized cutoff time $x_m$.
            \Ensure
            Estimation of the normalization factor $D_\tau$.
            \For{$p= 1~\text{\textbf{to}}~N_M$} 
                \State Sample a normalized evolution time $y_p$ from $\tilde{p}(y):= \frac{1}{2}\int_{-\infty}^{\infty}p(\frac{z+y}{2})p(\frac{z-y}{2})dz$.
                \If{$y_p>x_m$} \Comment{Exceed the preset cutoff value}
                    \State Set the estimation $\hat{d}_p=0$.
                \Else \Comment{Normal single-shot quantum sampling by Hadamard test}
                    \State Prepare an ancillary qubit on $\ket{+}$ state and the initial state $\ket{\psi_0}$.
                    \State Implement a controlled-unitary $C$-$U$ from the ancillary qubit to the state $\ket{\psi_0}$. Here, $U = e^{i\tau y_p H}$.
                    \State Generate a random bit $b_p$ with a uniform distribution of values $\{0,1\}$. Perform a gate $W^{b_p}$ on the ancillary qubit. Here, $S=\mr{diag}(1,-i)$ is a $\frac{\pi}{4}$-rotation gate.
                    \State Measure the ancillary qubit on $X$-basis, and record the binary result $a_p$. Then record the values $\{y_p,b_p,a_p\}$ and set the estimation value $\hat{d}_p:= 2(-1)^{a_p}\, i^{b_p} e^{-i\tau y_p E_j}$.
                \EndIf
            \EndFor
            \State Calculate the estimated normalization factor $\hat{D}^{(x_m)}_\tau:= \mr{Re}\left( \frac{1}{N_M} \sum_{p=1}^{N_M} \hat{d}_p \right)$.
        \end{algorithmic}
\end{algorithm}
\end{figure}

\subsection{Estimate the eigenenergy}

Suppose the initial state $\ket{\psi_0}$ has a constant overlap with the eigenstate $\ket{u_j}$. Then we can sweep the parameter $E_j^{(e)}$ in a range to maximize $D_\tau$. The estimation accuracy depends on the finite imaginary time $\tau$, finite truncation $x_m$, sampling number $N_M$ and the overlap $p_j$.

To clarify this, we consider a simple case where the $j$-th eigenenergy is know to be in a range $E_j \in [E_j^L, E_j^U]$; moreover, we suppose other eigenenergies is far from this range. We expand the initial state in the eigenstate basis,
\begin{equation}
\ket{\psi_0} = \sum_{j=0}^{N-1} c_j \ket{u_j}.
\end{equation}
We denote the square overlap of $\ket{\psi}$ and $\ket{u_i}$ to be
\begin{equation}
p_i = |\braket{u_i|\psi}|^2 = |c_i|^2.
\end{equation}

For an energy value $E$ in the range $[E_j^L,E_j^U]$, we calculate the ideal normalization factor $D_\tau$,
\begin{equation} \label{eq:DtaupsiVar}
\begin{aligned}
D_\tau(E) &= \braket{\psi| g^2(\tau(H-E)) |\psi} \\
&= g^2(\tau(E_j-E)) p_j   +  \sum_{i\neq j} g^2(\tau(E_i-E)) p_i \\
&\approx g^2(\tau(E_j-E)) p_j.
\end{aligned}
\end{equation}
The approximation holds when $g^2(\tau(E_j-E)) p_j \gg g^2(\tau(E_i-E)) p_i$. This naturally holds when the eigenenergies $\{E_i\}_{i\neq j}$ is far from the range $[E_j^L,E_j^U]$. From Eq.~\eqref{eq:DtaupsiVar} we can see that, the normalization factor takes the local maximum value close to $p_j$ when the energy value $E=E_j$,
\begin{equation}
E_j = \mathop{\mathrm{argmax}}\limits_{E\in[E_j^L,E_j^U]} D_\tau(E).
\end{equation}

Therefore, we can sweep the values of $E$ in a range $[E_j^L, E_j^U]$ to search the largest value of $D_\tau(E)$. The eigenenergy searching algorithm is in Algorithm~\ref{Alg:eigenenergyEst}. In the algorithm, we assume a free usage of computational resources. This can be improved by introducing better peak-value searching algorithm. We leave the improvement of the classical search process for future works.

\begin{figure}
\begin{algorithm}[H]
        \caption{Eigenenergy and normalization factor estimation}
        \label{Alg:eigenenergyEst}
        \begin{algorithmic}[1]
            \Require
            An $n$-qubit Hamiltonian $H$; initial state $\ket{\psi_0}$ with nonzero overlap with $j$-th eigenstate of $H$: $p_j = |\braket{u_j|\psi_0}|^2$; the energy interval $[E_j^L, E_j^U]$ for the $j$-th eigenstate; cooling function $g(h)$ and the corresponding sampling probability $p(x)$; proper choice of the imaginary time $\tau$, normalized cutoff time $x_m$.
            \Ensure
            Estimation of the eigenenergy $E_j$ and the corresponding normalization factor $\hat{D}^{(x_m)}_\tau$.
            \For{$p= 1~\text{\textbf{to}}~N_M$} 
                \State Sample a normalized evolution time $y_p$ from $\tilde{p}(y):= \frac{1}{2}\int_{-\infty}^{\infty}p(\frac{z+y}{2})p(\frac{z-y}{2})dz$.
                \If{$y_p>x_m$} \Comment{Exceed the preset cutoff value}
                    \State Set the estimation $\hat{d}_p=0$.
                \Else \Comment{Normal single-shot quantum sampling by Hadamard test}
                    \State Prepare an ancillary qubit on $\ket{+}$ state and the initial state $\ket{\psi_0}$.
                    \State Implement a controlled-unitary $C$-$U$ from the ancillary qubit to the state $\ket{\psi_0}$. Here, $U = e^{i\tau y_p H}$.
                    \State Generate a random bit $b_p$ with a uniform distribution of values $\{0,1\}$. Perform a gate $W^{b_p}$ on the ancillary qubit. Here, $S=\mr{diag}(1,-i)$ is a $\frac{\pi}{4}$-rotation gate.
                    \State Measure the ancillary qubit on $X$-basis, and record the binary result $a_p$. Then record the values $\{y_p,b_p,a_p\}$ and set the temporal estimation value $\hat{d}_p:= 2(-1)^{a_p}\, i^{b_p}$.
                \EndIf
            \EndFor
            \State Set the estimated normalization factor $\hat{D}^{(x_m)'}_\tau= 0$.
            \For{$E_j'$~\textbf{in}~$[E_j^L,E_j^U]$} \Comment{Try different possible eigenenergy value}
            \State Calculate the estimated normalization factor $\hat{D}^{(x_m)'}_\tau:= \mr{Re}\left( \frac{1}{N_M} \sum_{p=1}^{N_M} e^{-i\tau y_p E_j'} \hat{d}_p \right)$.
            \If{$\hat{D}^{(x_m)'}_\tau$}
                \State Set $\hat{D}^{(x_m)}_\tau =\hat{D}^{(x_m)'}_\tau$ and $E_j = E_j'$.
            \EndIf
            \EndFor
        \end{algorithmic}
\end{algorithm}
\end{figure}

\subsection{Estimate the unnormalized observable expectation value}

Now, we discuss the estimation $N_\tau(O)$ of a given observable 
\begin{equation}
O = \sum_{l\in\mc{P}_n} o_l P_l = \|O\|_1 \sum_{l\in\mc{P}_n} \mr{Pr}_O(l) P_l
\end{equation}
where $\mc{P}_n$ denotes the $n$-qubit Pauli group. $o_l$ is the coefficient of $O$ on the Pauli component $P_l$. We remark that, the coefficients $\{o_l\}$ are all set to be positive. The signs of the coefficients are put into the corresponding Pauli matrices $\{P_l\}$. Recall that $\|O\|_1$ is the $l_1$-norm of Pauli coefficients of $O$,
\begin{equation}
\|O\|_1 = \sum_{l\in\mc{P}_n} o_l.
\end{equation}
The probability distribution $\mr{Pr}_O(l)$ is defined to be,
\begin{equation}
\mr{Pr}_O(l) = \frac{o_l}{\sum_{l\in\mc{P}_n} o_l} = \frac{1}{\|O\|_1} o_l.
\end{equation} 

Recall the unbiased estimator of $N_\tau(O)$ is
\begin{equation}
\begin{aligned}
\hat{n} &= \|O\|_1 e^{-i\tau(x-x') E_j} \hat{r} = 2\|O\|_1 e^{-i\tau(x-x') E_j} i^b (-1)^a.
\end{aligned}
\end{equation}
If we randomly sample $x,x'$ from $p(x)$ and $l$ from $\mr{Pr}_O(l)$, based on Eq.~\eqref{eq:DNestimation} we have
\begin{equation}
\mbb{E}_{x,x',l,b,a} (\hat{n}) = N_\tau(O).
\end{equation}

We now introduce the estimation algorithm of $N_\tau(O)$ based on Eq.~\eqref{eq:DNestimationCut} and the Hadamard test. The unnormalized eigenstate property estimation is shown in Algorithm~\ref{Alg:UnnormEigenPropEst}.

\begin{figure}
\begin{algorithm}[H]
        \caption{Unnormalized eigenstate property estimation}
        \label{Alg:UnnormEigenPropEst}
        \begin{algorithmic}[1]
            \Require
            An $n$-qubit Hamiltonian $H$; initial state $\ket{\psi_0}$ with nonzero overlap with $j$-th eigenstate of $H$: $p_j = |\braket{u_j|\psi_0}|^2$; the energy $E_j$ for the $j$-th eigenstate; the target observable $O= \sum_{l\in\mbb{P}_n} o_l P_l$; cooling function $g(h)$ and the corresponding sampling probability $p(x)$; proper choice of the imaginary time $\tau$, normalized cutoff time $x_m$.
            \Ensure
            Estimation of the unnormalized observable value $N_\tau(O)$.
            \For{$q= 1~\text{\textbf{to}}~N_M$}  
                \State Sample two independent variables $x_q$ and $x_q'$ from $p(x)$.
                \If{$x_q>x_m$~\text{\textbf{or}}~$x_q'>x_m$} \Comment{Exceed the preset cutoff value}
                    \State Set the estimation $\hat{n}_q=0$.
                \Else \Comment{Normal single-shot quantum sampling by Hadamard test}
                    \State Randomly sample a Pauli string $P_l$ from the support of $O$, based on the probability distribution $\mr{Pr}_O(l)=o_l/\|O\|_1$. Here, $\|O\|_1 := \sum_{l} o_l$.
                    \State Prepare an ancillary qubit on $\ket{+}$ state and the initial state $\ket{\psi_0}$.
                    \State Implement a controlled-unitary $C$-$V$ from the ancillary qubit to the state $\ket{\psi_0}$. Here, $V = e^{i\tau x_q' H} P_l e^{i\tau x_q H}$.
                    \State Generate a random bit $b_q$ with a uniform distribution of values $\{0,1\}$. Perform a gate $W^{b_q}$ on the ancillary qubit.
                    \State Measure the ancillary qubit on $X$-basis, and record the binary result $a_q$. Then set the estimation value $\hat{n}_q= 2 \|O\|_1 (-1)^{a_q}\, i^{b_q}\, e^{-i\tau (x_q - x_q')E_j}$.
                \EndIf
            \EndFor
            \State Set the estimation of the unnormalized expectation value $\hat{N}^{(x_m)}_\tau(O):=\mr{Re}\left(\frac{1}{N_M}\sum_{q=1}^{N_M} \hat{n}_q \right)$.
        \end{algorithmic}
\end{algorithm}
\end{figure}

To summarize, we list the algorithm for the two applications in the main text, shown in Algorithms~\ref{Alg:GroundPropEst} and \ref{Alg:EigenPropEst}.

\begin{figure}
\begin{algorithm}[H]
        \caption{Ground state property estimation (with known ground-state energy)}
        \label{Alg:GroundPropEst}
        \begin{algorithmic}[1]
            \Require
            An $n$-qubit Hamiltonian $H$; initial state $\ket{\psi_0}$ with nonzero overlap with the ground state of $H$: $p_0 = |\braket{u_0|\psi_0}|^2$; the ground state energy $E_0$; the target observable $O= \sum_{l\in\mbb{P}_n} o_l P_l$; cooling function $g(h)$ and the corresponding sampling probability $p(x)$; proper choice of the imaginary time $\tau$, normalized cutoff time $x_m$.
            \Ensure
            Estimation of the observable value $\braket{u_0|O|u_0}$ for the ground state.
            \State Perform Algorithm~\ref{Alg:NormEst} to output the estimation of the normalization factor $\hat{D}^{(x_m)}_\tau$.
            \State Perform Algorithm~\ref{Alg:UnnormEigenPropEst} to output the estimation of the unnormalized observable value $\hat{N}^{(x_m)}_\tau $.
            \State Output the estimation $\braket{\hat{O}}^{(x_m)}_{\psi(\tau)} = \frac{\hat{N}^{(x_m)}_\tau(O)}{\hat{D}^{(x_m)}_\tau}$.
        \end{algorithmic}
\end{algorithm}
\end{figure}

\begin{figure}
\begin{algorithm}[H]
        \caption{Eigenstate property estimation (with unknown eigenstate energy)}
        \label{Alg:EigenPropEst}
        \begin{algorithmic}[1]
            \Require
            An $n$-qubit Hamiltonian $H$; initial state $\ket{\psi_0}$ with nonzero overlap with the ground state of $H$: $p_0 = |\braket{u_0|\psi_0}|^2$; the energy interval $[E^L_j, E^U_j]$ for the $j$-th eigenstate; the target observable $O= \sum_{l\in\mbb{P}_n} o_l P_l$; cooling function $g(h)$ and the corresponding sampling probability $p(x)$; proper choice of the imaginary time $\tau$, normalized cutoff time $x_m$.
            \Ensure
            Estimation of the observable value $\braket{u_j|O|u_j}$ for the $j$-th eigenstate.
            \State Perform Algorithm~\ref{Alg:eigenenergyEst} to output the estimation of the eigenenergy $E_j$ and the normalization factor $\hat{D}^{(x_m)}_\tau $.
            \State Perform Algorithm~\ref{Alg:UnnormEigenPropEst} to output the estimation of the unnormalized observable value $\hat{N}^{(x_m)}_\tau $.
            \State Output the estimation $\braket{\hat{O}}^{(x_m)}_{\psi(\tau)} = \frac{\hat{N}^{(x_m)}_\tau(O)}{\hat{D}^{(x_m)}_\tau }$.
        \end{algorithmic}
\end{algorithm}
\end{figure}

\section{Error and resource requirement analysis for the eigenenergy estimation} \label{Sec:eigenenergy_error_complexity}

In this section, we study the estimation error of the $j$-th eigenenergy using Algorithm~\ref{Alg:eigenenergyEst}. Based on the error dependence, we estimate the resource requirements (i.e., circuit depth and sample number) of the energy estimation. 

For simplicity, we focus on the case when the following assumptions holds,
\begin{enumerate}
\item The cooling function $g(h)$ is even and real. Therefore,
\begin{equation}
g(\tau(H-E)) = g^\dag(\tau(H-E)).
\end{equation}
\item There is only one eigenenergy in the range $[E^L_j, E^U_j]$. The value of $g(\tau(E-E_i))^2$ when $E_i\neq E_j$ is negligible.
\end{enumerate}

Under these assumptions, the normalization factor can then be expressed as,
\begin{equation} \label{eq:DtauEapprox}
\begin{aligned}
D_\tau(E) &= \braket{\psi_0|g(\tau(H-E))^2|\psi_0} \\
&= \sum_{j=0}^{N-1} g(\tau(E_j-E))^2 |\braket{u_j|\psi_0}|^2 \\
&= \sum_{j=0}^{N-1} g(\tau(E_j-E))^2 p_j \\
&\approx p_j g(\tau(E-E_j))^2.
\end{aligned}
\end{equation}
Here, the approximation holds when the values $g(\tau(E-E_i))^2$ contributed by other eigenenergies $E_i$ are negligible. 
In the following Proposition~\ref{prop:D_finitetau}, we will make this approximation rigorous. 
The location of the peak of $D_\tau(E)$ then indicates the eigenenergy $E_j$. 

In practice, however, we can only obtain the estimation $\hat{D}^{(x_m)}_\tau(E)$ of $D_\tau(E)$, considering the finite cutoff time $x_m$ and finite sample number $N_M$. We are going to prove that, the solution of the following maximization problem
\begin{equation}
\hat{E}_j := \argmax_{E\in[E^L_j, E^U_j]} \hat{D}^{(x_m)}_\tau(E),
\end{equation}
will be close to the real solution $E_j$.

In Fig.~\ref{fig:EnergyErrorSum} we summarize our analysis. 
In Sec.~\ref{Ssc:energyFinitetau}, we bound the distance between $D_\tau(E)$ and $p_j g(\tau(E-E_j))^2$ caused by finite $\tau$ when $E\in[E_j-\frac{\Delta}{2},E_j+\frac{\Delta}{2}]$. 
In Sec.~\ref{Ssc:energyFinitexm}, we bound the estimation error of the normalization factor $D_\tau(E)$ caused by the normalized cutoff time $x_m$.
In Sec.~\ref{Ssc:energyFiniteNM}, based on the measurement using Hadamard test, we bound the statistical error of the estimation of $D^{(x_m)}_\tau(E)$ caused by the finite sampling error $N_M$.
Finally, in Sec.~\ref{Ssc:energyFiniteSUM}, we consider the peak-value search problem using the estimator and analyze the eigenenergy accuracy dependence to the circuit depth and sample complexity. We will see how the cooling bandwidth $g^{-1}(1-\varepsilon)$ affects the accuracy of the peak-value search.

\begin{figure}[htbp]
\centering
\includegraphics[width=14cm]{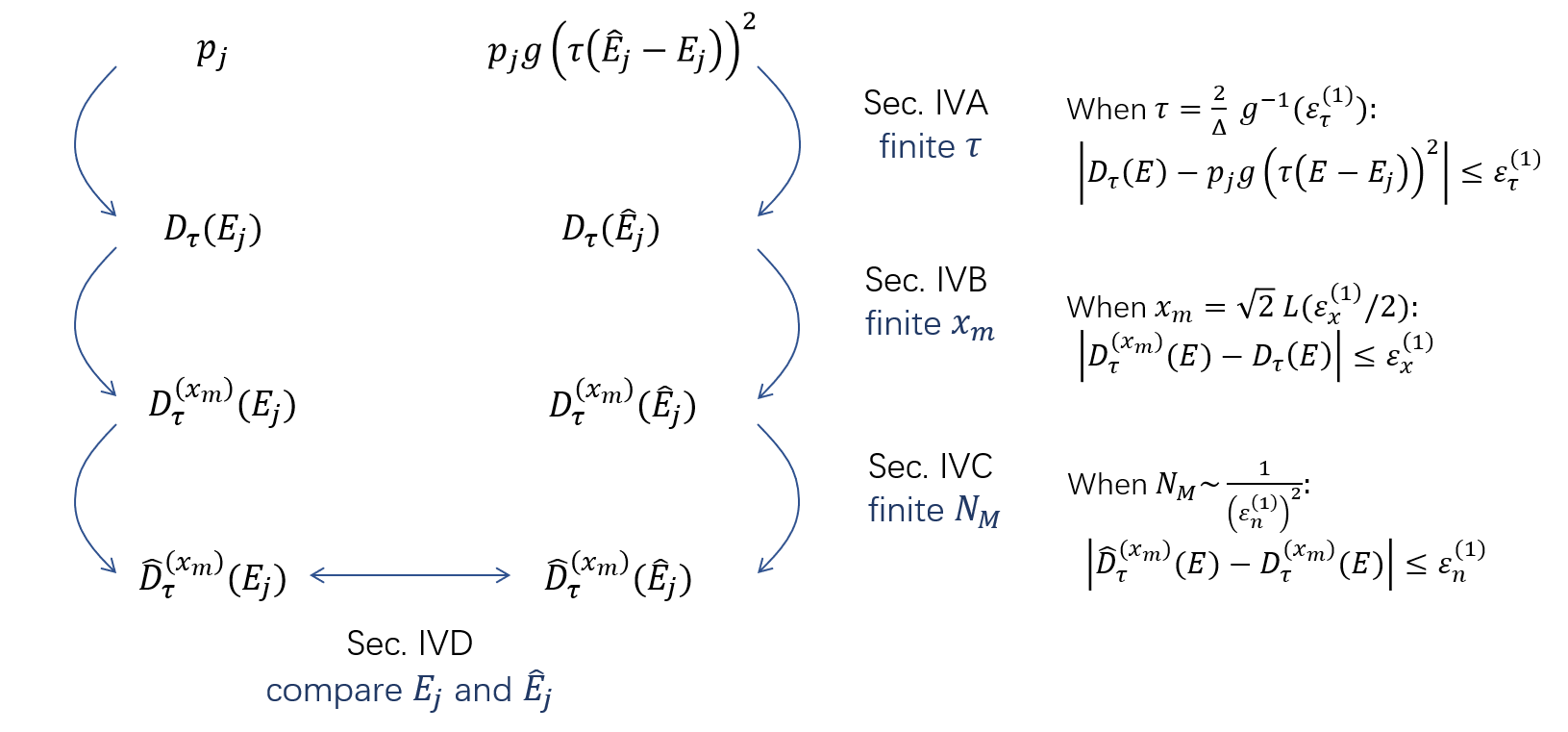}
\caption{Summary of the error analysis of eigenenergy estimation. We first study the effect of the finite imaginary time $\tau$, finite normalized cutoff time $x_m$ and finite sampling number $N_M$ on the normalization estimation, and then bound the difference between the estimation value $\hat{E}_j$ from true value $E_j$.} \label{fig:EnergyErrorSum}
\end{figure}

\subsection{Finite imaginary evolution time} \label{Ssc:energyFinitetau}

We first bound the distance between $D_\tau(E)$ and $p_j g(\tau(E-E_j))^2$ caused by finite $\tau$ when $E\in[E_j-\frac{\Delta}{2},E_j+\frac{\Delta}{2}]$.

\begin{proposition}[Error of the normalization function introduced by finite imaginary time] \label{prop:D_finitetau}
When $E\in[E_j-\frac{\Delta}{2},E_j+\frac{\Delta}{2}]$, we have
\begin{equation}
| D_\tau(E) - p_j g(\tau(E-E_j))^2 | \leq \varepsilon_\tau^{(1)},
\end{equation}
when $\tau \geq \frac{2}{\Delta} g^{-1}(\frac{\varepsilon_\tau^{(1)}}{2})$. When $E=E_j$, we can improve the requirement to $\tau\geq \frac{1}{\Delta} g^{-1}(\frac{\varepsilon_\tau^{(1)}}{2})$.
\end{proposition}

\begin{proof}
Denote the projector $\hat{P}_j = \ket{u_j}\bra{u_j}$. We first bound the distance between $g(\tau(H-E))$ and $g(\tau(E_j-E))\hat{P}_j$. Note that,
\begin{equation} \label{eq:gtauminPj}
g(\tau(H-E)) - g(\tau(E_j-E))\hat{P}_j = \sum_{i\neq j} g(\tau(E_i - E)) \ket{u_i}\bra{u_i}.
\end{equation}
When $E\in[E_j-\frac{\Delta}{2},E_j+\frac{\Delta}{2}]$, we have $|E-E_i|\geq \frac{\Delta}{2}$. In this case, when $\tau\geq \frac{2}{\Delta} g^{-1}(\frac{\varepsilon_\tau^{(1)}}{2})$, the coefficient $g(\tau(E_i-E))$ is bounded by
\begin{equation} \label{eq:gtauEiub}
\begin{aligned}
g(\tau(E_i-E)) &= g(\frac{2|E_i-E|}{\Delta} g^{-1}(\frac{\varepsilon_\tau^{(1)}}{2}) \\ 
&\leq g\left(g^{-1}(\frac{\varepsilon_\tau^{(1)}}{2}) \right) \\
&= \frac{\varepsilon_\tau^{(1)}}{2}.
\end{aligned}
\end{equation}
In the second line, we use the fact that $g(h')\leq g(h)$ when $h'>h>0$ in Definition~\ref{Def:coolingfunc}. When $E= E_j$, we have $|E_j-E_i|\geq \Delta$. In this case, we can achieve the same precision in Eq.~\eqref{eq:gtauEiub} using $\tau\geq \frac{1}{\Delta} g^{-1}(\frac{\varepsilon_\tau^{(1)}}{2})$.

Combining Eq.~\eqref{eq:gtauminPj} and Eq.~\eqref{eq:gtauEiub}, we have
\begin{equation} \label{eq:gtauHEminPjub}
\|g(\tau(H-E)) - g(\tau(E_j-E))\hat{P}_j\|_\infty \leq \frac{\varepsilon_\tau^{(1)}}{2}.
\end{equation}
When $\tau\geq \frac{2}{\Delta} g^{-1}(\frac{\varepsilon_\tau^{(1)}}{2})$. Here, $\|A\|_\infty$ denotes the spectral norm of matrix $A$. 

Now, we bound the distance between $D_\tau(E)$ and $p_j g(\tau(E-E_j))^2$,
\begin{equation}
\begin{aligned}
|D_\tau(E) - p_j g(\tau(E-E_j))^2| &= |\braket{\psi_0|g(\tau(H-E))g(\tau(H-E))|\psi_0} - g(\tau(E-E_j))^2 \braket{\psi_0|\hat{P}_j\hat{P}_j|\psi_0} | \\
&\leq |\braket{\psi_0|g(\tau(H-E))g(\tau(H-E))|\psi_0} - g(\tau(E-E_j)) \braket{\psi_0|g(\tau(H-E))\hat{P}_j|\psi_0} | \\
&\quad + g(\tau(E-E_j)) |\braket{\psi_0|g(\tau(H-E))\hat{P}_j|\psi_0} - g(\tau(E-E_j))\braket{\psi_0|\hat{P}_j\hat{P}_j|\psi_0} | \\
&\leq \frac{\varepsilon_\tau^{(1)}}{2} + g(\tau(E-E_j)) \frac{\varepsilon_\tau^{(1)}}{2} \\
&\leq \varepsilon_\tau^{(1)}.
\end{aligned}
\end{equation}
\end{proof}

\subsection{Finite normalized cutoff time} \label{Ssc:energyFinitexm}

We then bound the estimation error of $D_\tau(E)$ caused by the normalized cutoff time $x_m$. The estimation formula is given in Eq.~\eqref{eq:DNestimationCut},
\begin{equation} \label{eq:D_estimationCut}
\begin{aligned}
D_\tau^{(x_m)}(E) &= \int_{-x_m}^{x_m} dy \tilde{p}(y) e^{-i\tau y E_j} \braket{\psi_0| e^{i\tau y H} |\psi_0}.
\end{aligned}
\end{equation}
Here, $\tilde{p}(y) = \int_{-\infty}^{\infty}\, p(t-y) p(t) dt = [p\star p](y)$.

\begin{proposition}[Error of the normalization factor introduced by finite normalized cutoff time] \label{prop:D_finitexm}
The truncated estimation $D_\tau^{(x_m)}(E)$ defined in Eq.~\eqref{eq:D_estimationCut} is related to the normalization factor defined in Eq.~\eqref{eq:DNestimation} by
\begin{equation}
|D_\tau^{(x_m)}(E) - D_\tau(E)| \leq \varepsilon_x^{(1)}, \forall E\in\mbb{R},
\end{equation}
when the normalized cutoff time $x_m \geq \sqrt{2} L(\frac{\varepsilon_x^{(1)}}{2})$. Here, $L(\varepsilon)$ is the tail function of the cooling function $g(h)$ defined in Definition~\ref{Def:realizable}.
\end{proposition}

\begin{proof}
We have,
\begin{equation}
\begin{aligned}
|D_\tau^{(x_m)}(E) - D_\tau| &\leq (\int_{x_m}^{\infty} + \int_{-\infty}^{-x_m}) \,dy \left| \tilde{p}(y) e^{-i\tau y E_j} \braket{\psi_0| e^{i\tau y H} |\psi_0} \right| \\
&\leq (\int_{x_m}^{\infty} + \int_{-\infty}^{-x_m}) \,dy \tilde{p}(y) =: \varepsilon_x^{(1)}.
\end{aligned}
\end{equation}
Recall that
\begin{equation}
\begin{aligned}
\tilde{p}(y) &= \frac{1}{2}\int_{-\infty}^{\infty}\, p(\frac{z+y}{2})p(\frac{z-y}{2})dz \\
&= \int_{-\infty}^{\infty}\, p(t-y) p(t) dt = [p\star p](y).
\end{aligned}
\end{equation}
We can show the following relationship between $\tilde{p}(y)$ and $p(x)$ by the nonnegativity of $p(x)$,
\begin{equation}
\begin{aligned}
\int_{-x_m}^{x_m} \tilde{p}(y) dy &\geq \int_{-x_m/\sqrt{2}}^{x_m/\sqrt{2}} dx \int_{-x_m/\sqrt{2}}^{x_m/\sqrt{2}} dx\, p(x)p(x') \\
&\geq (1 -  L^{-1}(\frac{x_m}{\sqrt{2}}))^2 \\
&\geq 1 - 2 L^{-1}(\frac{x_m}{\sqrt{2}}),
\end{aligned}
\end{equation}
where $\varepsilon = L(x_m) = 1 - \int_{-x_m}^{x_m}p(x)dx$. 

Therefore, 
\begin{equation}
\begin{aligned}
& \varepsilon_x^{(1)} = 1 - \int_{-x_m}^{x_m} \tilde{p}(y) dy &\leq 2 L^{-1}(\frac{x_m}{\sqrt{2}}), \\
\Rightarrow \quad & x_m \geq \sqrt{2} L(\frac{\varepsilon_x^{(1)}}{2}).
\end{aligned}
\end{equation}

\end{proof}

\subsection{Finite sampling number} \label{Ssc:energyFiniteNM}

Now, we consider the statistical fluctuation when estimating the values of $D_\tau^{(x_m)}(E)$ using Eq.~\eqref{eq:D_estimationCut}. In the single-shot version of Hadamard test in Algorithm~\ref{Alg:eigenenergyEst}, we can describe the identical and identically distributed (i.i.d.) single-round estimators for $D_\tau^{(x_m)}(E)$ as $\{\hat{d}_p(E)\}_{p=1}^{N_M}$. Each single-round estimator $\hat{d}_p(E)$ is a random variable defined as follows,
\begin{equation}
\hat{d}_p =
\begin{cases}
\mr{Re}(2 \,e^{-i\tau y E}), & \quad (b,a)=(0,0), y:|y|\leq x_m\\
\mr{Re}(-2 \,e^{-i\tau y E}), & \quad (b,a)=(0,1), y:|y|\leq x_m \\
\mr{Re}(2i \,e^{-i\tau y E}), & \quad (b,a)=(1,0), y:|y|\leq x_m\\
\mr{Re}(-2i \,e^{-i\tau y E}), & \quad (b,a)=(1,1), y:|y|\leq x_m \\
0, & \quad y:|y|> x_m.
\end{cases}
\end{equation}

The final estimation of $D_\tau^{(x_m)}(E)$ is given by
\begin{equation} \label{eq:D_estimateNM}
\hat{D}_\tau^{(x_m)}(E) = \frac{1}{N_M} \sum_{p=1}^{N_M} \hat{d}_p(E).
\end{equation}
Based on Eq.~\eqref{eq:D_estimationCut}, we know that
\begin{equation}
\mbb{E}_{(y,b,a)}\left(\hat{D}_\tau^{(x_m)} \right) = D_\tau^{(x_m)}.
\end{equation}

To analyze the statistical fluctuation, we apply the Hoeffding bound.
\begin{lemma}[Hoeffding bound] \label{lem:Hoeffding} 
For $n$ independent random variables $\{\hat{X}_i\}_{i=1}^{n}$ which are bounded by $[a,b]$, the average value
\begin{equation}
\bar{X} := \frac{1}{n}\sum_{i=1}^{n} \hat{X}_i, 
\end{equation} 
satisfies
\begin{equation}
\Pr(\left|\bar{X} - \mbb{E}(\bar{X})\right|\geq \varepsilon) \leq 2 \exp\left( - \frac{2n \varepsilon^2}{(b-a)^2} \right).
\end{equation}
\end{lemma}

\begin{proposition}[Error of the normalization factor introduced by finite sampling number] \label{prop:D_finiteNM}
The estimator $\hat{D}^{(x_m)}_\tau(E)$ defined in Eq.~\eqref{eq:D_estimateNM} is related to the truncated estimation $\hat{D}^{(x_m)}_\tau(E)$ defined in Eq.~\eqref{eq:D_estimationCut} by
\begin{equation}
|\hat{D}_\tau^{(x_m)}(E) - D_\tau^{(x_m)}(E)| \leq \varepsilon_n^{(1)}, \forall E\in\mbb{R},
\end{equation}
with a failure probability $\delta^{(1)} := 2\exp\left( -2 K (\varepsilon^{(1)}_n)^2/16 \right)$, when the sample number $N_M \geq K (\varepsilon_n^{(1)})^{-2} $.
\end{proposition}

\begin{proof}
Note that, the estimators $\{\hat{d}_p(E)\}_{p=1}^{N_M}$ are independent, whose values are bounded by $[-2,2]$. 
Using the Hoeffding bound in Lemma~\ref{lem:Hoeffding} with the bound $[-2,2]$, we finish the proof.
\end{proof}

\subsection{Accuracy of the eigenenergy estimation} \label{Ssc:energyFiniteSUM}

Ideally, the eigenenergy $E_j$ satisfies,
\begin{equation} \label{eq:realEj}
E_j = \argmax_{E\in[E^L_j,E^U_j]} D_\tau(E).
\end{equation}
In practice, what we can solve is the following problem,
\begin{equation} \label{eq:hatEj}
\hat{E}_j = \argmax_{E\in[E^L_j,E^U_j]} \hat{D}^{(x_m)}_\tau(E).
\end{equation}

We now show that when $[E^L_j,E^U_j]\subset [E_j-\frac{\Delta}{2},E_j+\frac{\Delta}{2}]$, $\hat{E}_j$ is a good estimation of $E_j$ with precision $\kappa$ when the finite imaginary time $\tau=\mc{O}(\kappa^{-1})$, finite normalized cutoff time $x_m=\mc{O}(L(p_j))$, and sample number $N_M=\mc{O}(p_j^{-2})$.

\begin{theorem}[Accuracy of the eigenenergy estimation] \label{thm:eigenenergy}
When $[E^L_j,E^U_j]\subset [E_j-\frac{\Delta}{2},E_j+\frac{\Delta}{2}]$, the eigenenergy estimation $\hat{E}_j$ defined in Eq.~\eqref{eq:hatEj} is related to the eigenenergy $E_j$ by
\begin{equation}
|\hat{E}_j - E_j | \leq \kappa,
\end{equation}
with a failure probability of $2\delta^{(1)}= 4\exp\left(-K/8 \right)$, 
when the finite imaginary time $\tau\geq \frac{1}{\kappa}g^{-1}(\frac{1-g(1)}{6} p_j)$, the normalized cutoff time $x_m\geq \sqrt{2}L\left(\frac{1-g(1)}{6}p_j\right)$, and the sample number $N_M\geq \frac{9K}{(1-g(1))^2} p_j^{-2} $.
\end{theorem}

\begin{proof}
From Eq.~\eqref{eq:hatEj} we know that,
\begin{equation} \label{eq:DhatE_VS_DE}
\hat{D}^{(x_m)}_\tau(\hat{E}_j) \geq \hat{D}^{(x_m)}_\tau(E_j).
\end{equation}
Now, we take $x_m= \sqrt{2}L(\frac{\varepsilon_x p_j}{2})$ and $N_M = (\varepsilon_n p_j)^{-2}K$. Using Proposition~\ref{prop:D_finitexm} and \ref{prop:D_finiteNM}, we have
\begin{equation} \label{eq:Dhatxm2Dtau}
\begin{aligned}
| \hat{D}^{(x_m)}_\tau(\hat{E}_j) - D_\tau(\hat{E}_j) | \leq |D^{(x_m)}_\tau(\hat{E}_j) - D_\tau(\hat{E}_j)| + | \hat{D}^{(x_m)}_\tau(\hat{E}_j) - D^{(x_m)}_\tau(\hat{E}_j) | \leq p_j (\varepsilon_x + \varepsilon_n), \\
| \hat{D}^{(x_m)}_\tau(E_j) - D_\tau(E_j) | \leq |D^{(x_m)}_\tau(E_j) - D_\tau(E_j)| + | \hat{D}^{(x_m)}_\tau(E_j) - D^{(x_m)}_\tau(E_j) | \leq p_j (\varepsilon_x + \varepsilon_n) \\,
\end{aligned}
\end{equation}
each with a failure probability $\delta^{(1)}$. Furthermore, when $\tau= \frac{2}{\Delta}g^{-1}(\frac{p_j\varepsilon_\tau}{2})$, from Proposition~\ref{prop:D_finitetau} we have
\begin{equation} \label{eq:Dtau2D}
\begin{aligned}
|D_\tau(\hat{E}_j) - p_j g(\tau(E_j - \hat{E}_j))^2| &\leq p_j\varepsilon_\tau, \\
|D_\tau(E_j) - p_j| &\leq p_j\varepsilon_\tau.
\end{aligned}
\end{equation}

Combine Eqs.~\eqref{eq:DhatE_VS_DE}, \eqref{eq:Dhatxm2Dtau}, and \eqref{eq:Dtau2D}, we have
\begin{equation}
\begin{aligned}
p_j &\leq D_\tau(E_j) + \varepsilon_\tau \leq \hat{D}^{(x_m)}_\tau(E_j) + p_j(\varepsilon_x + \varepsilon_n) + p_j\varepsilon_\tau \\ 
& \leq \hat{D}^{(x_m)}_\tau(\hat{E}_j) + p_j(\varepsilon_x + \varepsilon_n) + p_j\varepsilon_\tau \\
& \leq D_\tau(\hat{E}_j) + 2p_j(\varepsilon_x + \varepsilon_n) + p_j\varepsilon_\tau, \\
& \leq p_j g(\tau(E_j - \hat{E}_j))^2 + 2p_j(\varepsilon_\tau + \varepsilon_x + \varepsilon_n), \\
\end{aligned}
\end{equation}
with a failure pribability $2\delta^{(1)}$. Therefore,
\begin{equation}
g(\tau (\hat{E}_j-E_j ) )^2 \geq 1 - 2(\varepsilon_\tau + \varepsilon_x + \varepsilon_n).
\end{equation}
Denote $\kappa:=\hat{E}_j-E_j$, we further simplify the expression above,
\begin{equation}
\begin{aligned}
\kappa &\leq \frac{1}{\tau} g^{-1}(\sqrt{1-2(\varepsilon_\tau +\varepsilon_x+\varepsilon_n)}) \\
&\leq \frac{1}{\tau} g^{-1}(1-(\varepsilon_\tau +\varepsilon_x+\varepsilon_n)) \\
&= \frac{1}{\tau} g^{-1}\left(1- \frac{2}{p_j}g(\frac{\Delta}{2}\tau) - \frac{2}{p_j} L^{-1}(\frac{x_m}{\sqrt{2}}) - \sqrt{ \frac{K}{p_j^2 N_M} } \right) \\
\end{aligned}
\end{equation}
In the second inequality, we use the property that $g^{-1}(p)$ is a decreasing function when $p>0$. Therefore, the following requirement is sufficient to make sure the eigenenergy error is smaller than $\kappa$,
\begin{equation}
\begin{aligned}
\tau &\geq \frac{1}{\kappa}, \\
\frac{2}{p_j}g(\frac{\Delta}{2}\tau) + \frac{2}{p_j} L^{-1}(\frac{x_m}{\sqrt{2}}) &+ \sqrt{ \frac{K}{p_j^2 N_M} } \leq 1 - g(1).
\end{aligned}
\end{equation}
Let $\frac{2}{p_j}g(\frac{\Delta}{2}\tau) = \frac{2}{p_j} L^{-1}(\frac{x_m}{\sqrt{2}}) = \sqrt{ \frac{K}{p_j^2 N_M} }$, we have
\begin{equation}
\begin{aligned}
\tau &\geq \frac{2}{\Delta} g^{-1}(\frac{1-g(1)}{6} p_j), \\
x_m &\geq \sqrt{2} L\left( \frac{1-g(1)}{6} p_j \right) \\
N_M &\geq K\left( \frac{3}{1-g(1)}\right)^2 \frac{1}{p_j^2}.
\end{aligned}
\end{equation}

In the above derivation, we require $\tau\geq \frac{1}{\kappa}$ and $\tau\geq \frac{2}{\Delta}g^{-1}(\frac{1-g(1)}{6} p_j)$. In practice, $\kappa \ll \frac{\Delta}{2}$. Then it suffices to have $\tau\geq \frac{1}{\kappa}g^{-1}(\frac{1-g(1)}{6} p_j)$.
\end{proof}

From Theorem~\ref{thm:eigenenergy} we can see that, the (maximal) circuit depth and sample complexity of Algorithm~\ref{Alg:eigenenergyEst} are $\tau x_m = \mc{O}(\kappa^{-1}g^{-1}(p_j)L(p_j))$ and $N_M =\mc{O}(p_j^{-2})$, respectively, for a given accuracy $\kappa$ and initial state overlap $p_j$. This achieves the Heisenberg limit $\frac{1}{\kappa}$ for the eigenenergy searching.

\section{Error and resource requirement analysis for the observable estimation} \label{Sec:cooling_error_complexity}

In this section, we study the estimation error of the observable value $\braket{O}$ in Algorithm~\ref{Alg:EigenPropEst} under the finite circuit depth and sampling number. Based on the error dependence, we estimate the resource requirements (i.e. circuit depth and sample number) of the cooling algorithm.

Similar to Sec.~\ref{Sec:eigenenergy_error_complexity}, we limit our analysis to the positive and even cooling function such that $g(h) = g(-h)$. For simplicity, we first assume that we have already obtained a precise eigenenergy estimation $E_j$ of the $j$-th eigenstate following Algorithm~\ref{Alg:eigenenergyEst}. 

In Fig.~\ref{fig:ErrorSum} we summarize our error analysis. In Sec.~\ref{Ssc:finitetau}, we consider the effect of finite imaginary time $\tau$; in Sec.~\ref{Ssc:finitexm}, we bound the estimation error caused by the normalized cutoff time $x_m$. The two factors above determine the actual maximum evolution time. In Sec.~\ref{Ssc:finiteNM}, based on the measurement using Hadamard test, we determine the statistical error caused by finite sampling number $N_M$. Finally, we summarize the finite sampling effect in Sec.~\ref{Ssc:finiteSum} and show the dependence of circuit depth $t_m$ and sampling number $N_M$ with respect to the initial state $\ket{\psi_0}$, Hamiltonian $H$ and the observable $O$.

Later on, in Sec.~\ref{Ssc:eigeneneryError}, we will discuss the effect of the eigenenergy estimation error $\kappa$ on the observable estimation. When $\kappa\ll \min\{|E_j - E_{j-1}|, |E_j - E_{j+1}|\}$, the error $\kappa$ will introduce negligible effect on the accuracy of the observable estimation. However, it will introduce a $\tau$-dependent factor on the normalization factor, which causes a higher sampling cost.

\begin{figure}[htbp]
\centering
\includegraphics[width=14cm]{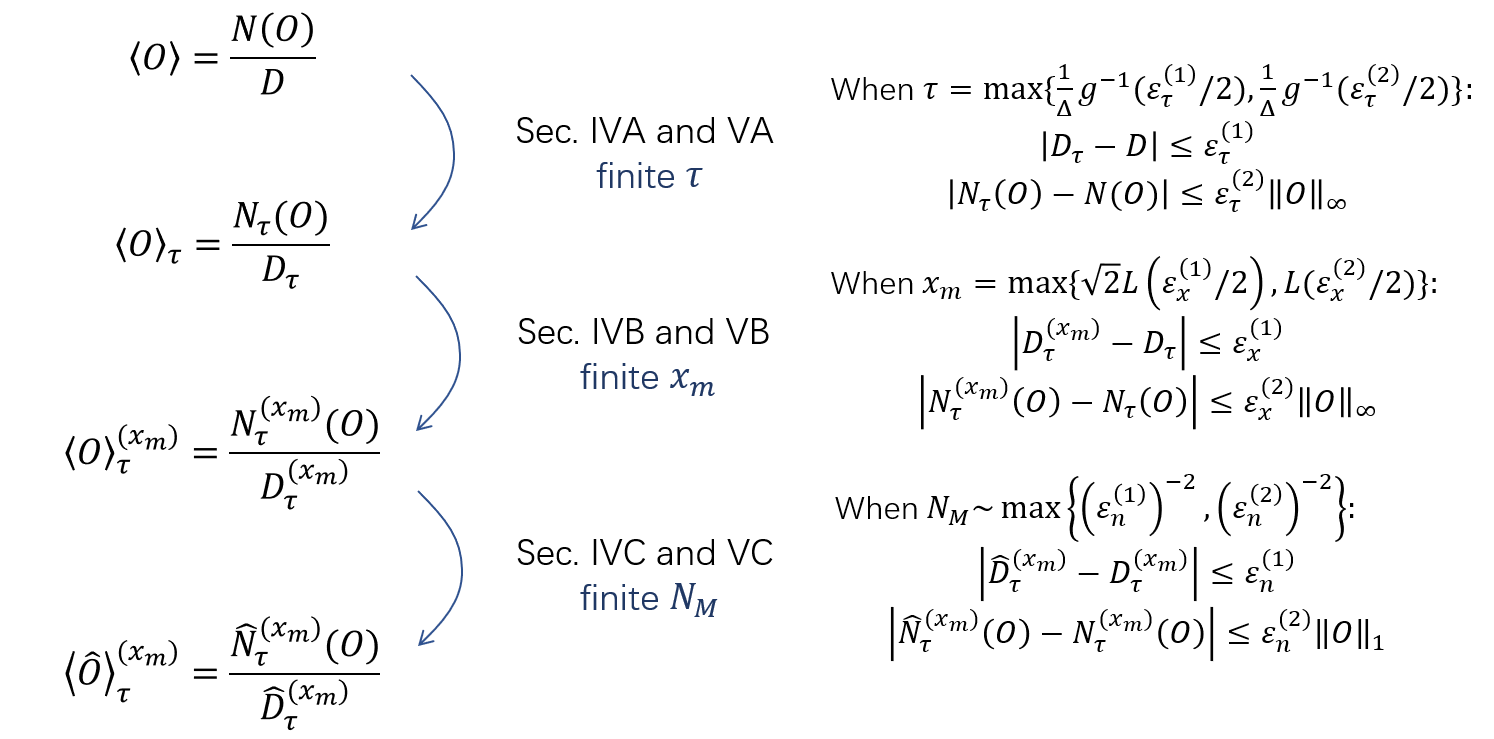}
\caption{Summary of the error analysis of the observable estimation. Started from the ideal observation value $\braket{O}$, we sequentially study the effect of finite imaginary time $\tau$, finite normalized cutoff time $x_m$ and finite sampling number $N_M$.} \label{fig:ErrorSum}
\end{figure}

\subsection{Finite imaginary evolution time} \label{Ssc:finitetau}

Recall that
\begin{equation} \label{eq:OexpPj}
\braket{O} = \braket{u_j|O|u_j} = \frac{N(O)}{D},
\end{equation}
where 
\begin{equation} \label{eq:DNO}
\begin{aligned}
D &= D(E_j) = \braket{\psi_0| \hat{P}_j |\psi_0} = p_j, \\
N(O) &= \braket{\psi|\hat{P}_j O \hat{P}_j |u_j} = p_j\braket{O}, \\
\end{aligned}
\end{equation}
are, respectively, the normalization factor and the unnormalized expectation value. Here, $\hat{P}_j:= \ket{u_j}\bra{u_j}$ is the projector to the $j$-th eigenstate. 

When the imaginary time $\tau$ is finite, we have the approximated value
\begin{equation} \label{eq:O_finiteTau}
\braket{O}_\tau = \frac{ N_\tau(O) }{ D_\tau },
\end{equation}
where
\begin{equation} \label{eq:DtauNtauO}
\begin{aligned}
D_\tau &= \braket{\psi_0| g^2(\tau(H-E_j)) |\psi_0}, \\
N_\tau(O) &= \braket{\psi_0| g(\tau(H-E_j))\, O \,g(\tau(H-E_j)) |\psi_0}.
\end{aligned}
\end{equation}
are, respectively, the normalization factor and the unnormalized expectation value with finite $\tau$.

In Proposition~\ref{prop:D_finitetau} in Sec.~\ref{Ssc:energyFinitetau}, we have analyzed the distance between $D_\tau(E)$ and $D(E)$ when $|E-E_j|\leq \frac{\Delta}{2}$. Similarly, we now bound the distance between $N_\tau(O)$ and $N(O)$.

\begin{proposition}[Error of the unnormalized expectation value introduced by finite imaginary time] \label{prop:N_finiteTau}
The unnormalized expectation value $N_\tau(O)$ defined in Eq.~\eqref{eq:DtauNtauO} is related to the ideal $N(O)$ defined in Eq.~\eqref{eq:DNO} by
\begin{equation}
|N_\tau(O) - N(O)| \leq \|O\|_\infty \varepsilon_\tau^{(2)},
\end{equation}
when the imaginary time $\tau\geq \frac{1}{\Delta}g^{-1}(\frac{\varepsilon_\tau^{(2)}}{2})$. Here, $\|O\|_\infty$ is the spectral norm of $O$.
\end{proposition}

\begin{proof}
To bound the distance between $N_\tau(O)$ and $N(O)$, we first bound the distance between $\hat{P}_j$ and $g(\tau(H-E_j))$. We have,
\begin{equation} \label{eq:gtauminPj2}
g(\tau(H-E_j)) - \hat{P}_j = \sum_{i\neq j} g(\tau(E_i-E_j)) \ket{u_i} \bra{u_i}.
\end{equation}
Denote $\Delta$ to be the known lower bound of $\min{\Delta_{j,j+1}, \Delta_{j,j-1}}$. If we choose $\tau$ to be $\tau \geq \frac{1}{\Delta} g^{-1}(\frac{1}{2} \varepsilon_\tau^{(2)})$, we have
\begin{equation} \label{eq:gtauEilb}
\begin{aligned}
g(\tau(E_i - E_j)) &= g\left( \frac{|E_i - E_j|}{\Delta} g^{-1}(\frac{1}{2}\varepsilon_\tau^{(2)}) \right) \\
&\leq g(g^{-1}(\frac{1}{2}\varepsilon_\tau^{(2)})) \\
&= \frac{1}{2}\varepsilon_\tau^{(2)}.
\end{aligned}
\end{equation}
In the second line, we use the definition of $g(h)$ that $g(h')\leq g(h)$ for all $h'>h>0$.
Based on Eq.~\eqref{eq:gtauminPj2} and Eq.~\eqref{eq:gtauEilb}, we have
\begin{equation} \label{eq:gtauminPjnorm}
\|g(\tau(H-E_j)) - \hat{P}_j \|_\infty \leq \frac{1}{2} \varepsilon_\tau^{(2)},
\end{equation}
when $\tau\geq \frac{1}{\Delta} g^{-1}(\frac{1}{2} \varepsilon_\tau^{(2)})$.

Now, we use Eq.~\eqref{eq:gtauminPjnorm} to bound the difference between $N_\tau(O)$ and $N(O)$. When $\tau\geq \frac{1}{\Delta} g^{-1}(\frac{1}{2} \varepsilon_\tau^{(2)})$, we then have,
\begin{equation}
\begin{aligned}
|N_\tau(O) - N(O)| &= | \braket{\psi_0| g(\tau(H-E_j)) O g(\tau(H-E_j)) |\psi_0} - \braket{\psi_0|\hat{P}_j O \hat{P}_j |\psi_0} | \\
&\leq  \left| \braket{\psi_0| g(\tau(H-E_j)) O g(\tau(H-E_j)) |\psi_0} - \braket{\psi_0|\hat{P}_j O g(\tau(H-E_j)) |\psi_0} \right| \\ 
& \quad + \left| \braket{\psi_0| \hat{P}_j O g(\tau(H-E_j)) |\psi_0} - \braket{\psi_0|\hat{P}_j O \hat{P}_j |\psi_0} \right| \\
&\leq 2\|g(\tau(H-E_j)) - \hat{P}_j\|_\infty \|O\|_\infty \\
&\leq \varepsilon_\tau^{(2)} \|O\|_\infty, \\
\end{aligned}
\end{equation}
\end{proof}

\subsection{Finite normalized cutoff time} \label{Ssc:finitexm}

In Proposition~\ref{prop:D_finitexm}, we bound the estimation error of $D_\tau(O)$ caused by finite $x_m$. Follow similar methods, we bound the estimation errors of $N_\tau(O)$ caused by finite $x_m$. The estimation formula is given in Eq.~\eqref{eq:DNestimationCut},
\begin{equation} \label{eq:N_finitexm}
\begin{aligned}
N_\tau^{(x_m)}(O) &= \|O\|_1 \int_{-x_m}^{(x_m)} dx \int_{-x_m}^{x_m} dx' p(x) p(x') \sum_{l\in\mc{P}_n} \mr{Pr}_O(l) e^{-i\tau (x-x') E_j} \braket{\psi_0| e^{-i\tau x' H} P_l e^{i\tau x' H} |\psi_0},
\end{aligned}
\end{equation}
where 
\begin{equation} 
O = \sum_{l\in\mc{P}_n} o_l P_l = \|O\|_1 \sum_{l\in\mc{P}_n} \mr{Pr}_O(l) P_l, 
\end{equation}
$\|O\|_1$ is the $l_1$-norm of Pauli coefficients of $O$,
\begin{equation}
\|O\|_1 = \sum_{l\in\mc{P}_n} o_l.
\end{equation}
Note that, the coefficients $\{o_l\}_l$ are all positive, since we put the signs into the Pauli matrices $\{P_l\}$. The probability distribution $\mr{Pr}_O(l)$ is defined to be,
\begin{equation}
\mr{Pr}_O(l) = \frac{o_l}{\sum_{l\in\mc{P}_n} o_l} = \frac{1}{\|O\|_1} o_l.
\end{equation}

\begin{proposition}[Error of the unnormalized expectation value introduced by finite normalized cutoff time] \label{prop:N_finitexm}
The truncated estimation $N_\tau^{(x_m)}(O)$ defined in Eq.~\eqref{eq:N_finitexm} is related to the normalization factor $N_\tau(O)$ defined in Eq.~\eqref{eq:DtauNtauO} by
\begin{equation}
|N_\tau^{(x_m)}(O) - N_\tau(O)| \leq \|O\|_{\infty} \varepsilon_x^{(2)},
\end{equation}
when the normalized cutoff time $x_m \geq L(\frac{\varepsilon_x^{(2)}}{2})$. Here, $L(\varepsilon)$ is the tail function of the cooling function $g(h)$ defined in Definition~\ref{Def:realizable}.
\end{proposition}

\begin{proof}
We have,
\begin{equation}
\begin{aligned}
|N_\tau^{(x_m)}(O) - N_\tau(O)| &\leq \iint_{\bar{S}} dx dx'\, p(x)p(x') \left| e^{-i\tau (x-x') E_j} \braket{\psi_0| e^{-i\tau x' H} O e^{i\tau x' H} |\psi_0} \right| \\
&\leq \iint_{\bar{S}} dx dx'\, p(x)p(x') \left| \braket{\psi_0| e^{-i\tau x' H} O e^{i\tau x' H} |\psi_0} \right| \\
&\leq \|O\|_{\infty} \iint_{\bar{S}} dx dx'\, p(x)p(x') =: \|O\|_{\infty} \varepsilon_x^{(2)},
\end{aligned}
\end{equation}
where $\bar{S}$ implies the complement area of $S:\{(x,x')| |x|\leq x_m, |x'|\leq x_m \}$.

The error term $\varepsilon_x^{(2)}$ is related to the tail probability by
\begin{equation}
\begin{aligned}
&\varepsilon_x^{(2)} = 1 - \iint_{S} dx dx'\, p(x)p(x') = 1 - (1-L^{-1}(x_m))^2 \leq 2L^{-1}(x_m) \\
\Rightarrow\quad & x_m \geq L(\frac{\varepsilon_x^{(2)}}{2}).
\end{aligned}
\end{equation}

Therefore, when we choose the $x_m$ to be $L(\varepsilon_x^{(2)})$, we can achieve the estimation of $N_\tau(O)$ with a precesion of $\|O\|_{\infty} \varepsilon_x^{(2)}$.

\end{proof}


\subsection{Finite sample number} \label{Ssc:finiteNM}

In Proposition~\ref{prop:D_finitexm}, we bound the estimation error of $D^{(x_m)}_\tau(O)$ caused by finite sampling error $N_M$. Follow similar methods, we bound the statistical error of $N^{(x_m)}_\tau(O)$.

In the single-shot version of Hadamard test in Algorithm~\ref{Alg:EigenPropEst}, we can describe these identical and identically distributed (i.i.d.) single-round estimators for $N_\tau^{(x_m)}(O)$ as $\{\hat{n}_q\}_{q=1}^{N_M}$. Each single-round estimator $\hat{n}_q$ is a random variable defined as follows,
\begin{equation}
\hat{n}_q(O) =
\begin{cases}
\mr{Re}(2\|O\|_1 \,e^{-i\tau (x-x') E_j} ), & \quad (b,a)=(0,0), (x,x'):|x|\leq x_m, |x'|\leq x_m \\
\mr{Re}(-2\|O\|_1 \,e^{-i\tau (x-x') E_j} ), & \quad (b,a)=(0,1), (x,x'):|x|\leq x_m, |x'|\leq x_m \\
\mr{Re}(2i \|O\|_1 \,e^{-i\tau (x-x') E_j} ), & \quad (b,a)=(1,0), (x,x'):|x|\leq x_m, |x'|\leq x_m\\
\mr{Re}(-2i \|O\|_1 \,e^{-i\tau (x-x') E_j} ), & \quad (b,a)=(1,1), (x,x'):|x|\leq x_m, |x'|\leq x_m \\
0, & \quad (x,x'):|x|> x_m \text{ or } |x'|>x_m.
\end{cases}
\end{equation}

The final estimation of $N_\tau^{(x_m)}(O)$ is given by
\begin{equation} \label{eq:N_estimateNM}
\begin{aligned}
\hat{N}_\tau^{(x_m)}(O) &= \frac{1}{N_M} \sum_{q=1}^{N_M} \hat{n}_q(O).
\end{aligned}
\end{equation}

Based on Eq.~\eqref{eq:N_finitexm}, we know that
\begin{equation}
\begin{aligned}
\mbb{E}_{(x,x',P,b,a)} \left( \hat{N}_\tau^{(x_m)}(O) \right) &= N_\tau^{(x_m)}(O).
\end{aligned}
\end{equation}

Now we bound the statistical error of $\hat{N}_\tau^{(x_m)}(O)$ using Lemma~~\ref{lem:Hoeffding}. 

\begin{proposition}[Error of the observable expectation value introduced by finite sampling number] \label{prop:N_finiteNM}
The estimator $\hat{N}^{(x_m)}_\tau(O)$ defined in Eq.~\eqref{eq:N_estimateNM} is related to the truncated estimation $\hat{N}^{(x_m)}_\tau(O)$ defined in Eq.~\eqref{eq:N_finitexm} by
\begin{equation}
|\hat{N}_\tau^{(x_m)}(O) - N_\tau^{(x_m)}(O)| \leq \|O\|_1\varepsilon_n^{(2)},
\end{equation}
with a failure probability $\delta^{(2)} := 2\exp\left( -2 K (\varepsilon^{(2)}_n)^2/16 \right)$, when the sample number $N_M \geq K (\varepsilon_n^{(2)})^{-2} $. Here, $\|O\|_1 = \sum_l |o_l|$ is the sum of the Pauli coefficients of $O$.
\end{proposition}

\begin{proof}
Note that, the estimators $\{\hat{n}_q(O)\}_{q=1}^{N_M}$ are independent, whose values are bounded by $[-2\|O\|_1,2\|O\|_1]$. Here, $\|O\|_1 = \sum_l |o_l|$. 
Using the Hoeffding bound in Lemma~\ref{lem:Hoeffding} with the bound $[-2\|O\|_1, 2\|O\|_1]$, we finish the proof.
\end{proof}

\subsection{Accuracy of the observable estimation} \label{Ssc:finiteSum}

The final observable estimation is given by
\begin{equation} \label{eq:OestComb}
\braket{\hat{O}}^{(x_m)}_\tau = \frac{\hat{N}^{(x_m)}_\tau}{\hat{O}^{(x_m)}_\tau}.
\end{equation}

In Propositions~\ref{prop:D_finitetau}, \ref{prop:D_finitexm}, \ref{prop:D_finiteNM}, \ref{prop:N_finiteTau}, \ref{prop:N_finitexm}, and \ref{prop:N_finiteNM}, we estimate the errors caused by finite $\tau$, $x_m$, and $N_M$ on $D_\tau$ and $N_\tau(O)$, respectively. Now we combine the results together.

\begin{proposition} \label{prop:Ocomb}
For the initial state $\ket{\psi_0}$, target eigenstate $\ket{u_j}$ and cooling function $g(h)$, if we set the imaginary time $\tau$, normalized cutoff time $x_m$, and sample number $N$ to be
\begin{enumerate}
\item $\tau= \frac{1}{\Delta} \max\{ g^{-1}(\frac{\varepsilon_\tau^{(1)}}{2}, g^{-1}(\frac{\varepsilon_\tau^{(2)}}{2} ) \}$;
\item $x_m= \max\{\sqrt{2}L(\frac{\varepsilon_x^{(1)}}{2}), L(\frac{\varepsilon_x^{(2)}}{2})\}$;
\item $N_M= \max\{K(\varepsilon^{(1)}_n)^{-2}, K(\varepsilon^{(2)}_n)^{-2}\}$;
\end{enumerate}
then the estimated observable expectation value $\braket{\hat{O}}^{(x_m)}_\tau$ defined in Eq.~\eqref{eq:OestComb} in Algorithm~\ref{Alg:EigenPropEst} is related to the real observable expectation value $\braket{O}$ by, 
\begin{equation}
\begin{aligned}
|\braket{\hat{O}}_{\tau}^{(x_m)} - \braket{O} | &\leq p_j^{-1} (\braket{O}+1)( \varepsilon^{(1)}_{\tau} + \varepsilon^{(1)}_{x} + \varepsilon^{(1)}_{n} ) + p_j^{-1} (\varepsilon^{(2)}_{\tau} \|O\|_{\infty} + \varepsilon^{(2)}_{x} \|O\|_{\infty} + \varepsilon^{(2)}_{n} \|O\|_1),
\end{aligned}
\end{equation}
with a failure probability $\delta^{(1)}+\delta^{(2)}=4e^{-\frac{K}{2}}$. 
\end{proposition}

\begin{proof}
Following Propositions~\ref{prop:D_finitetau}, \ref{prop:D_finitexm}, \ref{prop:D_finiteNM}, \ref{prop:N_finiteTau}, \ref{prop:N_finitexm}, and \ref{prop:N_finiteNM}, when we choose $\tau$, $x_m$, and $N_M$ to be the values mentioned above, then the following bound holds at the same time
\begin{equation} \label{eq:DNfinite_xmNM}
\begin{aligned}
\left| \hat{D}_\tau^{(x_m)}  - D \right| &\leq \varepsilon^{(1)}_{\tau} + \varepsilon^{(1)}_{x} + \varepsilon^{(1)}_{n}, \\
\left| \hat{N}_\tau^{(x_m)}(O) - N(O) \right| &\leq (\varepsilon^{(2)}_{\tau} \|O\|_{\infty} + \varepsilon^{(2)}_{x} \|O\|_{\infty} + \varepsilon^{(2)}_{n} \|O\|_1 ) ,
\end{aligned}
\end{equation}
with a failure probability $\delta(K)= 4 e^{-\frac{K}{8}}$. To simplify the notation, we denote $\varepsilon^{(1)}:= \varepsilon^{(1)}_{\tau} + \varepsilon^{(1)}_{x} + \varepsilon^{(1)}_{n}$ and $\varepsilon^{(2)}(O):= \varepsilon^{(2)}_{\tau} \|O\|_{\infty} + \varepsilon^{(2)}_{x} \|O\|_{\infty} + \varepsilon^{(2)}_{n} \|O\|_1$.

Based on Eq.~\eqref{eq:DNfinite_xmNM}, we bound the difference between the quotients $\braket{\hat{O}}_{\tau}^{(x_m)}$ and $\braket{O}$,
\begin{equation} \label{eq:Ofinite_xmNM}
\begin{aligned}
|\braket{O} - \braket{\hat{O}}_{\tau}^{(x_m)}| &= \left| \frac{N(O)}{D} - \frac{ \hat{N}_\tau^{(x_m)}(O) }{ \hat{D}_\tau^{(x_m)} } \right|  \\
&= \left| \frac{N(O)\hat{D}_\tau^{(x_m)} -  D \hat{N}_\tau^{(x_m)}(O)   }{ D \hat{D}_\tau^{(x_m)} } \right| \\
&\leq \left| \frac{N(O) (D + \varepsilon^{(1)})  -  D( N(O)-\varepsilon^{(2)}(O) )   }{ D (D - \varepsilon^{(1)}) } \right|  \\
&= \left| \frac{N(O) \varepsilon^{(1)} + D \varepsilon^{(2)}(O) )   }{D^2 - D\varepsilon^{(1)}} \right| \\
&\leq \frac{(N(O)+D)\varepsilon^{(1)} + D \varepsilon^{(2)}(O) )   }{D^2} \\
&= p_j^{-1} (\braket{O}+1)( \varepsilon^{(1)}_{\tau} + \varepsilon^{(1)}_{x} + \varepsilon^{(1)}_{n} ) + p_j^{-1} (\varepsilon^{(2)}_{\tau} \|O\|_{\infty} + \varepsilon^{(2)}_{x} \|O\|_{\infty} + \varepsilon^{(2)}_{n} \|O\|_1).
\end{aligned}
\end{equation}
In the fifth line, we use the fact that $N(O)\varepsilon^{(1)}$ and $D\varepsilon^{(2)}(O)$ are much smaller than $D^2=p_j^2$.

\end{proof}

Based on Proposition~\ref{prop:Ocomb}, we can estimate the circuit depth and sample complexity for the observable estimation task.

\begin{theorem}[Accuracy of the observable estimation] \label{thm:observable} The expectation value estimation $\braket{\hat{O}}^{(x_m)}_\tau$ defined in Eq.~\eqref{eq:OestComb} is related to the real observable expectation value $\braket{O}$ by
\begin{equation}
\begin{aligned}
|\braket{\hat{O}}_{\tau}^{(x_m)} - \braket{O} | &\leq \varepsilon\left( \frac{1}{2}(\braket{O}+1) + \frac{1}{3}\|O\|_\infty + \frac{1}{6}\|O\|_1 \right) \\
&\leq \varepsilon (\|O\|_1 + 1)
\end{aligned}
\end{equation}
with a failure probability of $\delta^{(1)}+\delta^{(2)}=4 \exp(-K/8)$, when the finite imaginary time $\tau\geq \frac{1}{\Delta}g^{-1}\left(\frac{\varepsilon p_j}{12} \right)$, the normalized cutoff time $x_m\geq \sqrt{2} L\left(\frac{\varepsilon p_j}{12}\right)$, and the sample number $N_M \geq K\, (\frac{\varepsilon p_j}{6})^{-2}$. Here, $\|O\|_1 = \sum_l |o_l|$ is the sum of the Pauli coefficients of $O$, $\Delta$ is a known lower bound of $\min\{|E_j-E_{j-1}|,|E_j-E_{j+1}|\}$.
\end{theorem}

\begin{proof}
From Proposition~\ref{prop:Ocomb} we know that, if we set $\tau=\frac{1}{\Delta}\max\{g^{-1}(\frac{\varepsilon^{(1)}_\tau}{2}), g^{-1}(\frac{\varepsilon^{(2)}_\tau}{2}) \}$  $x_m=\max\{\sqrt{2}L(\frac{\varepsilon_x^{(1)}}{2}), L(\frac{\varepsilon_x^{(2)}}{2})\}$ and $N_M=\max\{K(\varepsilon^{(1)}_n)^{-2}, K(\varepsilon^{(2)}_n)^{-2}\}$, then 
\begin{equation} \label{eq:OexpminOreal}
\begin{aligned}
|\braket{O} - \braket{\hat{O}}_{\tau}^{(x_m)}| &\leq p_j^{-1} (\braket{O}+1)( \varepsilon^{(1)}_{\tau} + \varepsilon^{(1)}_{x} + \varepsilon^{(1)}_{n} ) + p_j^{-1} (\varepsilon^{(2)}_{\tau} \|O\|_{\infty} + \varepsilon^{(2)}_{x} \|O\|_{\infty} + \varepsilon^{(2)}_{n} \|O\|_1) \\
&\leq p_j^{-1}(\|O\|_1 + 1)( \varepsilon^{(1)}_{\tau} + \varepsilon^{(1)}_{x} + \varepsilon^{(1)}_{n}) + p_j^{-1}\|O\|_1(\varepsilon^{(2)}_{\tau} + \varepsilon^{(2)}_{x} + \varepsilon^{(2)}_{n})) \\
\end{aligned}
\end{equation}

If we set
\begin{equation} \label{eq:splitepsilon}
\begin{aligned}
\varepsilon^{(1)}_\tau = \varepsilon^{(2)}_\tau = \varepsilon^{(1)}_x = \varepsilon^{(2)}_x = \varepsilon^{(1)}_n = \varepsilon^{(2)}_n &= \frac{\varepsilon p_j}{6},
\end{aligned}
\end{equation}
which corresponds to 
\begin{equation}
\begin{aligned}
\tau \geq \frac{1}{\Delta} g^{-1}\left( \frac{\varepsilon p_j}{12}\right) \\
x_m \geq \sqrt{2} L\left( \frac{\varepsilon p_j}{12} \right), \\
N_M \geq K\, \left(\frac{\varepsilon p_j}{6}\right)^{-2},
\end{aligned}
\end{equation}
then Eq.~\eqref{eq:OexpminOreal} becomes,
\begin{equation}
\begin{aligned}
|\braket{\hat{O}}_{\tau}^{(x_m)} - \braket{O} | &\leq \varepsilon\left( \frac{1}{2}(\braket{O}+1) + \frac{1}{3}\|O\|_\infty + \frac{1}{6}\|O\|_1 \right) \\
&\leq \varepsilon (\|O\|_1 + 1).
\end{aligned}
\end{equation}
\end{proof}

From Theorem~\ref{thm:observable} we can see that, the circuit depth and sample complexity of the observable estimation are $\tau x_m = \mc{O}(\frac{1}{\Delta}g^{-1}\left( \varepsilon p_j \right) L(\varepsilon p_j))$ and $N_M = \mc{O}(\varepsilon^{-2} p_j^{-2})$, respectively. 

As is shown in Eqs.~\eqref{eq:tri_inv},~\eqref{eq:exp_inv},~\eqref{eq:gau_inv}, and \eqref{eq:sech_inv}, the common cooling functions satisfy
\begin{equation}
g^{-1}(\varepsilon) = \mc{O}\left(\log(\varepsilon^{-1})\right),
\end{equation}
that is, with a logarithm increasing $\tau$, we can suppress the error term caused by finite imaginary time evolution. On the other hand, for the common cooling functions,
\begin{equation}
L(\varepsilon) = \mc{O}(\mr{poly}(\varepsilon^{-1})).
\end{equation}
This implies that the circuit depth of a general cooling function is $\mc{O}(\mr{poly}(\varepsilon^{-1})) \log(\varepsilon^{-1})$. For the Gaussian cooling or secant hyperbolic function we can further reduce it to $\mc{O}(\mr{poly}\left(\log(\varepsilon^{-1})\right))$, which is exponentially better than the phase estimation algorithm.

\subsection{Effect of eigenenergy estimation error on the observable estimation} \label{Ssc:eigeneneryError}

Now, we analyze the effect of the energy estimation error $\kappa$. As a reasonable assumption, we assume that $|\kappa|<\frac{\Delta}{2}$,i.e., the estimated eigenenergy $\hat{E}_j$ satisfies $|\hat{E}_j - E_j|\leq \frac{\Delta}{2}$.

The normalization factor $D_\tau$ and unnormalized observation value $N_\tau(O)$ will become
\begin{equation} \label{eq:DtauNtauOkappa}
\begin{aligned}
D_\tau(\hat{E}_j) &= \braket{\psi_0| g^2(\tau(H-\hat{E}_j)) |\psi_0}, \\
N_\tau(O;\hat{E}_j) &= \braket{\psi_0| g(\tau(H-\hat{E}_j))\, O \,g(\tau(H-\hat{E}_j)) |\psi_0}.
\end{aligned}
\end{equation}

In practice, we will generate estimated values $\hat{D}_\tau^{(x_m)}(\hat{E}_j)$ and $\hat{N}_\tau^{(x_m)}(O;\hat{E}_j)$ following similar methods introduced in Sec.~\ref{Ssc:finitexm} and Sec.~\ref{Ssc:finiteNM}. The final estimation of $\braket{O}$ will be
\begin{equation} \label{eq:OestKappa}
\braket{\hat{O}}_\tau^{(x_m)}(\hat{E}_j) = \frac{\hat{N}_\tau^{(x_m)}(O;\hat{E}_j)}{\hat{D}_\tau^{(x_m)}(\hat{E}_j)}. 
\end{equation}

To bound the distance between $\braket{\hat{O}}_\tau^{(x_m)}(\hat{E}_j)$ and $\braket{O}$, we only need to relate $\hat{N}_\tau^{(x_m)}(O;\hat{E}_j)$ and $\hat{D}_\tau^{(x_m)}(\hat{E}_j)$ to $N(O)$ and $D$, respectively. The errors introduced by finite $x_m$ and $N_M$ can be still analyzed based on Proposition~\ref{prop:D_finitexm}, \ref{prop:D_finiteNM}, \ref{prop:N_finitexm}, and \ref{prop:N_finiteNM}. Furthermore, the finite $\tau$ error in the normalization factor $D(E)$ can be handled by Proposition~\ref{prop:D_finitetau}.
Therefore, we only need to generalize Proposition~\ref{prop:N_finiteTau}, following similar analysis in Proposition~\ref{prop:D_finitetau}.

\begin{proposition}[Error of the unnormalized expectation value introduced by finite imaginary time with energy estimation inaccuracy] \label{prop:N_finiteTauKappa}
When $E\in[E_j-\frac{\Delta}{2},E_j+\frac{\Delta}{2}]$, the finite $\tau$ unnormalized expectation value $N_\tau(O)$ defined in Eq.~\eqref{eq:DtauNtauOkappa} is related to the ideal $N(O)$ defined in Eq.~\eqref{eq:DNO} by
\begin{equation}
|N_\tau(O;E) - g(\tau(E-E_j))^2 N(O)| \leq \|O\|_\infty \varepsilon_\tau^{(2)},
\end{equation}
when the imaginary time $\tau\geq \frac{1}{\Delta}g^{-1}(\frac{\varepsilon_\tau^{(2)}}{2})$. Here, $\|O\|_\infty$ is the spectral norm of $O$.
\end{proposition}

\begin{proof}

From Eq.~\eqref{eq:gtauHEminPjub} in the proof in Proposition~\ref{prop:D_finitetau} we have
\begin{equation} \label{eq:gtauHEminPjub2}
\|g(\tau(H-E)) - g(\tau(E_j-E))\hat{P}_j\|_\infty \leq \frac{\varepsilon_\tau^{(2)}}{2},
\end{equation}
when $\tau\geq \frac{2}{\Delta}g^{-1}(\frac{\varepsilon_\tau^{(2)}}{2})$.

Now, we use Eq.~\eqref{eq:gtauHEminPjub2} to bound the difference between $N_\tau{O;E}$ and $g(\tau(E-E_j))^2 N(O)$,
\begin{equation}
\begin{aligned}
|N_\tau(O;E) - g(\tau(E-E_j))^2 N(O)| &= | \braket{\psi_0| g(\tau(H-E_j)) O g(\tau(H-E_j)) |\psi_0} - g(\tau(E-E_j))^2 \braket{\psi_0|\hat{P}_j O \hat{P}_j |\psi_0} | \\
&\leq  \left| \braket{\psi_0| g(\tau(H-E_j)) O g(\tau(H-E_j)) |\psi_0} - g(\tau(E-E_j)) \braket{\psi_0|\hat{P}_j O g(\tau(H-E_j)) |\psi_0} \right| \\ 
& \quad + g(\tau(E-E_j)) \left| \braket{\psi_0| \hat{P}_j O g(\tau(H-E_j)) |\psi_0} - g(\tau(E-E_j)) \braket{\psi_0|\hat{P}_j O \hat{P}_j |\psi_0} \right| \\
&\leq \|g(\tau(H-E_j)) - \hat{P}_j\|_\infty \|O\|_\infty + g(\tau(E-E_j)) \|g(\tau(H-E_j)) - \hat{P}_j\|_\infty \|O\|_\infty \\
&\leq 2\|g(\tau(H-E_j)) - \hat{P}_j\|_\infty \|O\|_\infty \\
&\leq \varepsilon_\tau^{(2)} \|O\|_\infty.
\end{aligned}
\end{equation}
\end{proof}

Based on Propositions~\ref{prop:D_finitetau}, \ref{prop:D_finitexm}, \ref{prop:D_finiteNM}, \ref{prop:N_finitexm}, \ref{prop:N_finiteNM}, and \ref{prop:N_finiteTauKappa}, we then have the following proposition, which is a modified version of Proposition~\ref{prop:Ocomb}.

\begin{proposition} \label{prop:OcombKappa}
For the initial state $\ket{\psi_0}$, target eigenstate $\ket{u_j}$ and cooling function $g(h)$, when the finite energy estimation error satisfies $|\kappa|\leq \frac{\Delta}{2}$, if we set the imaginary time $\tau$, normalized cutoff time $x_m$, and sample number $N$ to be
\begin{enumerate}
\item $\tau= \frac{2}{\Delta} \max\{ g^{-1}(\frac{\varepsilon_\tau^{(1)}}{2}, g^{-1}(\frac{\varepsilon_\tau^{(2)}}{2} ) \}$;
\item $x_m= \max\{\sqrt{2}L(\frac{\varepsilon_x^{(1)}}{2}), L(\frac{\varepsilon_x^{(2)}}{2})\}$;
\item $N_M= \max\{K(\varepsilon^{(1)}_n)^{-2}, K(\varepsilon^{(2)}_n)^{-2}\}$;
\end{enumerate}
then the estimated observable expectation value $\braket{\hat{O}}^{(x_m)}_\tau(\hat{E}_j)$ defined in Eq.~\eqref{eq:OestComb} in Algorithm~\ref{Alg:EigenPropEst} is related to the real observable expectation value $\braket{O}$ by, 
\begin{equation}
\begin{aligned}
|\braket{\hat{O}}_{\tau}^{(x_m)} - \braket{O} | &\leq p_j^{-1} g(\tau\kappa)^{-2} (\braket{O}+1)( \varepsilon^{(1)}_{\tau} + \varepsilon^{(1)}_{x} + \varepsilon^{(1)}_{n} ) + p_j^{-1} g(\tau\kappa)^{-2} (\varepsilon^{(2)}_{\tau} \|O\|_{\infty} + \varepsilon^{(2)}_{x} \|O\|_{\infty} + \varepsilon^{(2)}_{n} \|O\|_1),
\end{aligned}
\end{equation}
with a failure probability $\delta^{(1)}+\delta^{(2)}=4e^{-\frac{K}{2}}$. 
\end{proposition}

The proof of Proposition~\ref{prop:OcombKappa} is omitted since it is almost the same as the one in Proposition~\ref{prop:Ocomb}. We only need to replace $N(O)$ and $D$ to $g(\tau\kappa)^2 N(O)$ and $g(\tau\kappa)^2 D$, respectively.

Proposition~\ref{prop:OcombKappa} indicates that, when the eigenenergy is estimated to a precision of $\kappa$, the observation error will become $g(\tau\kappa)^2$ times larger than the case when $E_j$ is accurate. This will put an upper bound on the value of $\tau$: when $\tau\to\infty$, $g(\tau\kappa)^2\to 0$. 

Similar to Theorem~\ref{thm:observable}, we have the following proposition.

\begin{proposition}[Accuracy of the observable estimation with small eigenenergy estimation error] \label{prop:observableKappa} When the finite estimation error satisfies $|\kappa|\leq \frac{\Delta}{2}$, the expectation value estimation $\braket{\hat{O}}^{(x_m)}_\tau(\hat{E}_j)$ defined in Eq.~\eqref{eq:OestKappa} is related to the real observable expectation value $\braket{O}$ by
\begin{equation}
\begin{aligned}
|\braket{\hat{O}}_{\tau}^{(x_m)} - \braket{O} | & \leq \varepsilon\left( \frac{1}{2}(\braket{O}+1) + \frac{1}{3}\|O\|_\infty + \frac{1}{6}\|O\|_1 \right) \\
& \leq \varepsilon(\|O\|_1 +1 ),
\end{aligned}
\end{equation}
with a failure probability of $\delta^{(1)}+\delta^{(2)}=4 \exp(-K/8)$, when the finite imaginary time $\tau$ satisfies
\begin{enumerate}
    \item $\tau\leq \tau_m$, where $\tau_m=\mc{O}(1/\kappa)$ so that $g(\tau_m\kappa)=\Omega(1)$;
    \item $\tau\geq \frac{2}{\Delta}g^{-1}\left(\frac{\varepsilon p_j}{12} \right)$,
\end{enumerate} 
the normalized cutoff time $x_m\geq \sqrt{2} L\left(\frac{\varepsilon p_j g(\tau\kappa)^2}{12}\right)$, and the sample number $N_M \geq K\, (\frac{\varepsilon p_j g(\tau\kappa)^2}{6})^{-2}$. Here, $\|O\|_1 = \sum_l |o_l|$ is the sum of the Pauli coefficients of $O$.
\end{proposition}

The proof is essentially the same as the one in Theorem~\ref{thm:observable}. We only need to modify Eq.~\eqref{eq:splitepsilon} to
\begin{equation}
\varepsilon^{(1)}_\tau = \varepsilon^{(2)}_\tau = \varepsilon^{(1)}_x = \varepsilon^{(2)}_x = \varepsilon^{(1)}_n = \varepsilon^{(2)}_n = \frac{\varepsilon p_j}{6} g(\tau\kappa)^2,
\end{equation}
so that Proposition~\ref{prop:OcombKappa} holds.

Now we discuss the requirement of energy estimation precision $\kappa$ based on Proposition~\ref{prop:observableKappa}. To make sure that the observable $O$ can be accurately estimated, the finite imaginary evolution time $\tau$ must satisfy the following two requirements,
\begin{enumerate}
    \item $\tau\leq \tau_m$, where $\tau_m=\mc{O}(1/\kappa)$ so that $g(\tau_m\kappa)=\Omega(1)$;
    \item $\tau\geq \frac{2}{\Delta}g^{-1}\left(\frac{\varepsilon p_j}{12} \right)$.
\end{enumerate}
As a result, the value of $\kappa$ should satisfy
\begin{equation}
    \kappa \leq \mc{O}\left(\frac{\Delta}{g^{-1}(\frac{\epsilon p_j}{6})} \right).
\end{equation}
So that we are able to estimate the observable values with proper $\tau$.

\section{Numerical simulation}

We show the spectrum search with different total evolution time $T = \tau x_m$ with the cutoff $x_m$ and imaginary evolution time $\tau$. Here,  $\tau$ and $x_m$ are calculated according to the precision, as discussed in the above sections. The maximum imaginary evolution time $\tau$ ranges from $\tau = 0.9$ to $1.7$. As shown in Figure.~\ref{fig:Spectra_Tau}, with increasing imaginary time, we can distinguish the peaks that are close to each other, aligning with our theoretical analysis. 

In the main text, both the imaginary time $\tau$ and the cutoff $x_m$ are consistently adjusted by the target precision.
We elaborate how to determine the the imaginary time $\tau$ and the cutoff $x_m$ by taking the Gaussian function as an example.
For a certain simulation accuracy $\varepsilon$, the imaginary time $\tau$ and the cutoff $x_m$ satisfy
$\tau \propto \sqrt{ \log(2\Delta/\epsilon)}$ and $x_m = 2\sqrt{\log (2/\epsilon)}$, respectively.
The reference line in the main text shows the error due to finite time $\tau$  plus the theoretical maximum error $L(\epsilon)$.


We also show the simulation accuracy under the cooling process. 
Similarly in the main text, suppose we aim to find the second eigenstate $\ket{u_2}$, which has the largest overlap with the initial state. Given the associated eigenenergy $E_2$ found by the above method, we now analyze the simulation error during the cooling process
Here, the observable that verifies the state infidelity as $O = \ket{u_2}\bra{u_2}$. We track the simulation accuracy by comparing with the ideal expectation value under the generalized cooling process, which is given by $\braket{\psi_0|g(\tau  (H-E_2) )^{\dagger} O g(\tau (H-E_2) )|\psi_0}$.
We can see that the error is kept in a relatively small value during the cooling process. We can improve the simulation accuracy by increasing the  Monte Carlo sampling of the integral.

\begin{figure}
\centering
\includegraphics[width=0.5\linewidth]{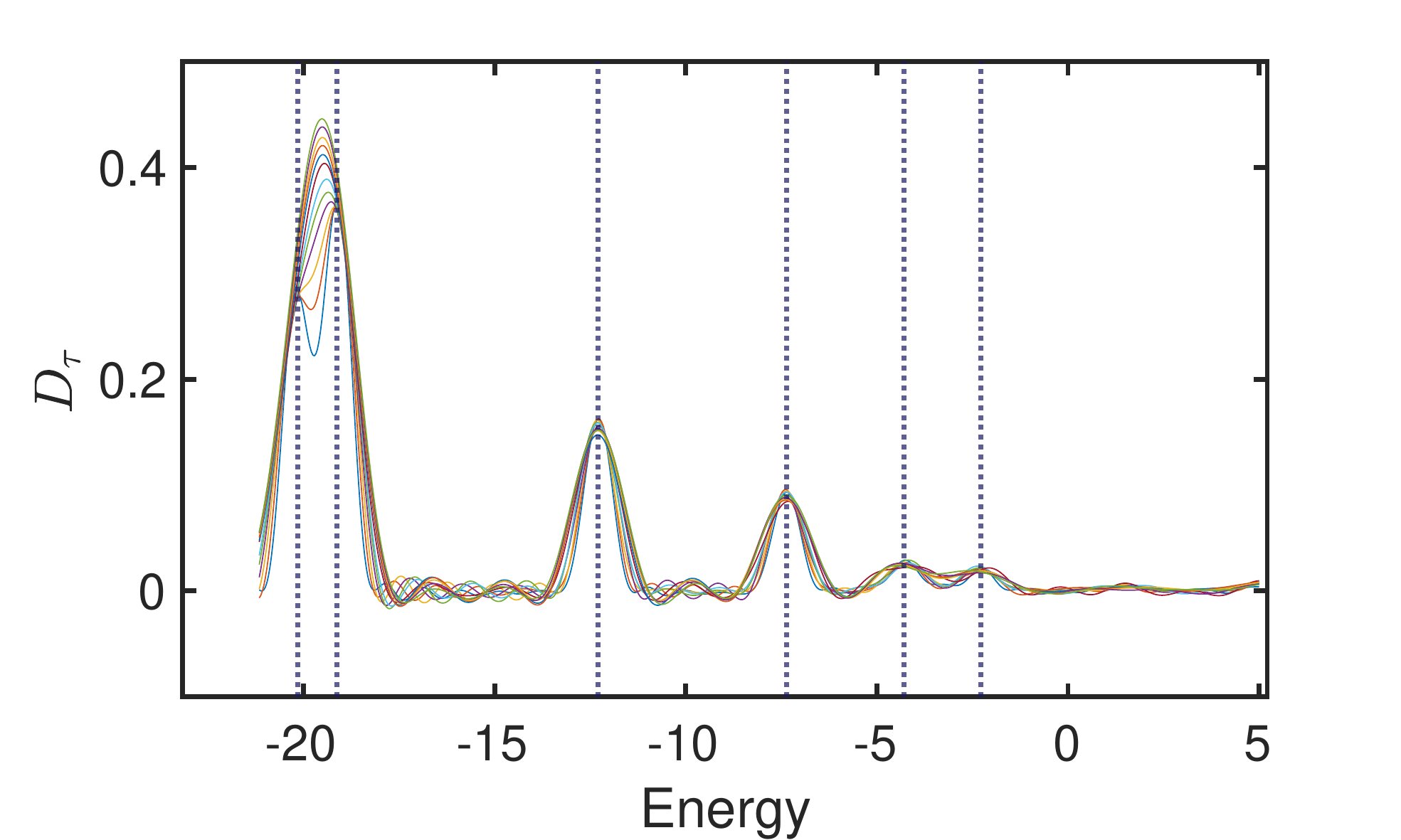}
\caption{ We show the spectrum search of the $8$-site anisotropic Heisenberg model using the Gaussian cooling function with different evolution time. The solid line shows the denominator of $D_{\tau}$ with different $\tau$ from $\tau = 0.9$ to $1.7$. The cutoff is calculated according to the precision in the main text.
We set $N_s = 10^5$ in the Monte Carlo sampling of the integral, and we use the expectation value for each sample that ignores measurement shot noise to keep consistent as that in the main text.}
 \label{fig:Spectra_Tau}
\end{figure}

\begin{figure}
\centering
  \includegraphics[width=0.5\linewidth]{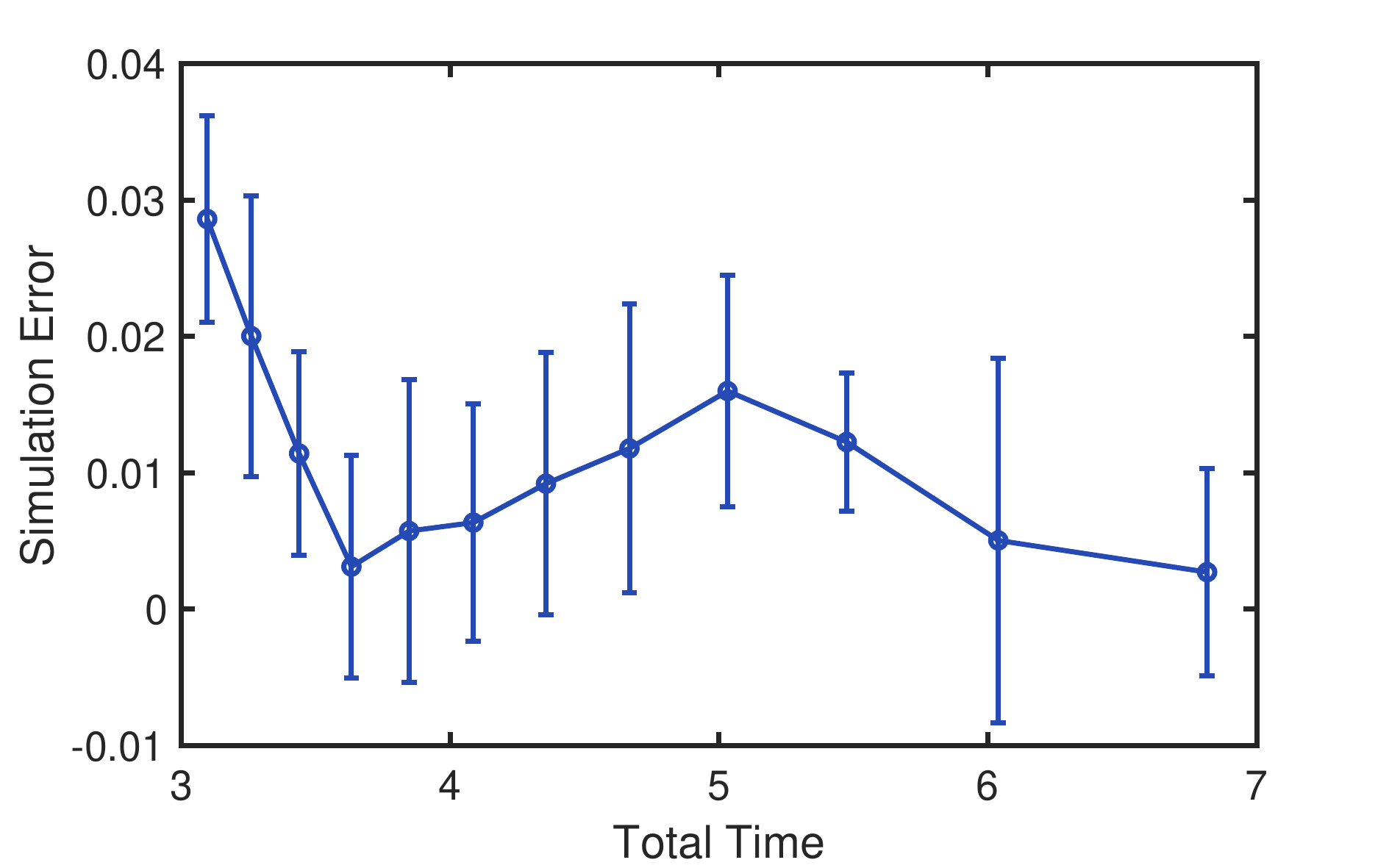}
  \caption{ Simulation error during the cooling process. The error is compared with  the ideal expectation value given by $\braket{\psi_0|g(\tau  (H-E_2) )^{\dagger} O g(\tau  (H-E_2) )|\psi_0}$.
  The simulation error is below $0.05$ under evolution.}
  \label{fig:Error_tau}
\end{figure}


\end{document}